\documentclass[a4paper,10pt]{article}

\usepackage{amstext}
\usepackage{amssymb,amsmath}
\usepackage{epstopdf}
\usepackage{tikz}

\usepackage[utf8]{inputenc} 
\usepackage[T1]{fontenc}    
\usepackage{hyperref}       
\usepackage{url}            
\usepackage{booktabs}       
\usepackage{amsfonts}       
\usepackage{nicefrac}       
\usepackage{microtype}      
\usepackage{lipsum}		
\usepackage{graphicx}
\usepackage{natbib}
\usepackage{doi}
\usepackage{authblk}
\usepackage[title]{appendix}

\usepackage{bm}
\usepackage{subfig}

\newtheorem{Proposition}{Proposition}
\newtheorem{Theorem}{Theorem}

\newtheorem{Definition}{Definition}
\newtheorem{proof}{Proof}
\newtheorem{Remark}{Remark}
\newtheorem{Corollary}{Corollary}

\newcommand{\CC}{\mathbb{C}}

\newcommand{\E}{\mathbb{E}}

\newcommand{\X}{\mathbf{X}}

\newcommand{\U}{\mathbf{U}}
\newcommand{\V}{\mathbf{V}}

\newcommand{\0}{\mathbf{0}}

\newcommand{\x}{\mathbf{x}}

\newcommand{\LL}{\bm{L}}
\newcommand{\DD}{\bm{D}}

\newcommand{\A}{\bm{A}}
\newcommand{\B}{\bm{B}}
\newcommand{\C}{\bm{C}}
\newcommand{\I}{\bm{I}}
\newcommand{\K}{\bm{K}}
\newcommand{\M}{\bm{M}}

\newcommand{\OO}{\bm{O}}

\newcommand{\Siga}{\bm{\Sigma}}
\newcommand{\Dada}{\bm{\Delta}}

\newcommand{\muv}{\bm{\mu}}

\newcommand{\SSS}{\bm{S}}

\newcommand{\diag}{\mathrm{diag}}

\newcommand{\Var}{\mathrm{Var}}

\newcommand{\pa}{\, \mathrm{par}}

\newcommand{\indep}{\perp\!\!\!\!\perp} 

\let\svthefootnote\thefootnote
\newcommand\blankfootnote[1]{%
  \let\thefootnote\relax\footnotetext{#1}%
  \let\thefootnote\svthefootnote%
}
\title{Causal Vector Auto-Regression Enhanced with Covariance and Order Selection}

\makeatletter
\let\@fnsymbol\@arabic
\makeatother
\author{%
Marianna Bolla\thanks{
Department of Stochastics, Budapest University of Technology and Economics, Hungary; marib@math.bme.hu
}
$^{,*}$,
Dongze Ye\thanks{
{Department of Computer Science, University of Southern California, USA}; dongzeye@usc.edu
} ,
Haoyu Wang\thanks{
{Department of Computational and Applied Mathematics, University
of Chicago, USA}; haoyuwang@uchicago.edu
} ,
Renyuan Ma\thanks{
{Department of Statistics, Yale University, USA}; jack.ma.rm2545@yale.edu
} ,
Valentin Frappier\thanks{{UFR Sciences and Techniques, Nantes University, France}; valentin.frappier@gmail.com
} , 
William Thompson\thanks{
{Lindner College of Business, University of Cincinnati, USA}; thompson.williamfn@gmail.com
} ,
Catherine Donner\thanks{
{Data Science and Analytics Institute, University of Oklahoma in
Norman, USA}
} ,
M\'at\'e Baranyi\thanks{
{Department of Stochastics, 
Budapest University of Technology and Economics, Hungary}; baranyim@math.bme.hu
}
, and 
Fatma Abdelkhalek\thanks{
{Department of Statistics, Mathematics, and Insurance, Faculty of Commerce, Assiut University, Egypt}; fatma.said@aun.edu.eg
}
}

\date{}

\hypersetup{
pdftitle={Causal Vector Auto-Regression Enhanced with Covariance and Order Selection},
pdfauthor={Bolla et al.},
pdfkeywords={Structural vector Auto-Regression; Causality along a DAG; Block Cholesky decomposition; Covariance selection; Order selection},
}

\begin{document}
\maketitle
\begin{abstract}
A causal vector autoregressive (CVAR) model is introduced for weakly stationary multivariate processes, combining a recursive directed graphical model for the contemporaneous components and a vector autoregressive model longitudinally. Block Cholesky decomposition with varying block sizes is used to solve the model equations and estimate the path coefficients along a directed acyclic graph (DAG). If the DAG is decomposable, i.e. the zeros form a reducible zero pattern (RZP) in its adjacency matrix,  then  covariance selection is applied that assigns zeros to the corresponding path coefficients.
Real life applications are also considered, where for the optimal order \texorpdfstring{$p\ge 1$}{p≥1} of  the fitted CVAR\texorpdfstring{$(p)$}{(p)} model, order selection is performed with various information criteria.
\end{abstract}
\blankfootnote{\hspace{-4pt}*Correspondence: marib@math.bme.hu; Tel.: +36 1 2000646}%
  
\noindent
\textbf{Keywords:} Structural vector Auto-Regression; Causality along a DAG; Block Cholesky decomposition; Covariance selection; Order selection

\noindent\textbf{PACS:} J0101  \quad \textbf{MSC:} 15B05; 62M10; 65F99; 65F30

\section{Introduction}\label{intro}

The purpose of this paper is to connect graphical modeling tools and time series models together via path coefficient estimation. In statistics, the path analysis was established by the geneticist~\cite{Wright} about a century
ago, but he used complicated entrywise calculations with partial correlations. Taking these partial correlations of a pair of variables in a multidimensional data set conditioned on another set of variables, makes things overtly complicated, as the conditioning set changes in every step. A bit later, in econometrics, structural equation modeling (SEM) was developed;
the prominent author~\cite{Haavelmo} obtained the Nobel price for it later. The maximum likelihood estimation (MLE) of the parameters in the Gaussian case  was elaborated by~\cite{Joreskog}. At the same time, \cite{Wold1}, inventor of partial least squares regression (PLS), used matrix calculations, and~\cite{Kiiveri} already used block matrix decompositions when dividing their variables into endogenous and
exogenous ones. However, none of these authors applied steadily
algorithms of block LDL (variant of the Cholesky)  decomposition
alone, without using partial correlations. Further, they did not consider time series. 

Here we  give a rigorous block matrix approach of these problems originated in statistics and time series analysis.
Furthermore, we enhance the usual and structural vector
autoregressive (VAR and SVAR) models, discussed e.g. in~\cite{Deistler} and~\cite{Deistlerbook}, with a causal component that has an effect between the coordinates contemporaneously. Therefore, we call our causal vector autoregressive  model CVAR.
Joint effects between the contemporaneous components are also considered in SVAR models of~\cite{Keating} and~\cite{Lutkepohl,Kilian}, but just recursive ordering of the variables, and no specific structure of the underlying directed graph is investigated. 
Though in~\cite{Wold} a causal chain model is introduced with an exogenous and a lagged endogenous variable, Gaussian Markov processes and usual regression estimates are used in context of econometrical problems.
This research is also  inspired by the paper~\cite{Wermuth}, 
where recursive ordering of the variables is crucial, without using any time component. 

\cite{Eichler1} introduces causality as a fundamental tool for  the empirical investigations of dynamic interactions in multivariate time series. He also discusses the differences between the structural and Granger causality. Former one appears in the SVAR models (see~\cite{Geweke}), whereas the latter one first appears in~\cite{Wiener}, then in~\cite{Granger}, and is sometimes called Wiener--Granger causality.
Without causality between the contemporaneous components, our model in the Gaussian case also resembles the one of~\cite{Eichler}, where the error term (shock) can have correlated components. Our higher order recursive VAR model can be transformed into a model like this, but the price is loosing
the recursive structure. The VAR model of~\cite{Brillinger} has a similar structure as ours with uncorrelated error terms; but no
further benefits of recursive models, like RZP, induced by the underlying DAG, are discussed.
\cite{Sims} investigates the use of different type structural equation and autoregressive models in macroeconomy,  
without suggesting  numerical algorithms.
However, historically this survey paper was among the first ones which pointed out the difference between the so far existing macroeconomic models and distinguished endogenous and exogenous variables. The method of the most recent paper~\cite{Bazinas} is based on the reduced form system and is constructed by the conditional distribution of two endogenous variables, given a catalyst or multiple catalysts; lagged effects are assessed, without having a longer time series, and stationarity is not assumed. 

Throughout the paper, second order processes are considered that can be assumed to follow multivariate Gaussian distribution.
In Section~\ref{methods}, the different types of VAR models are compared, and  a novel  causal VAR model is introduced,
combining the above recursive model contemporaneously and a VAR$(p)$ model longitudinally. In Section~\ref{results}, the models are described in details, together with introducing algorithms for the parameter estimation.
In Subsection~\ref{model}, the unrestricted CVAR(1) model is introduced, while in Subsection~\ref{incomplete}, the restricted cases are treated, with some prescribed zeros in the path coefficients. Relation to covariance selection and decomposability is discussed too.
In Subsection~\ref{high}, higher order CVAR models are introduced. 
In Section~\ref{appl}, application to real life data is presented together with information criteria for order selection (optimal choice of $p$).
The results and estimation schemes are summarized in Section~\ref{discussion};
finally, in Section~\ref{concl}, conclusions and further 
perspectives are posed. The proofs of the theorems and the detailed description of the algorithms are to be found in Appendix~\ref{appA}, 
while a survey about graphical
models, with emphases to the Gaussian case, in Appendix~\ref{appB}.

\section{Materials \& Methods}\label{methods}

First the different purpose VAR$(p)$ models for the $d$-dimensional,
weakly stationary process $\{ \X_t \}$ are enlisted and compared. The first two
models are known in the literature, whereas the last two ones are our
contributions, for which block matrix decomposition based algorithms are 
introduced in Section~\ref{results}, and they are illustrated in 
Section~\ref{appl} on real life data.

\begin{itemize}
\item
\textbf{Reduced form} VAR$(p)$ model: for given integer $p\ge 1$, it is
\begin{equation}\label{reduced}
 \X_{t} + \M_1 \X_{t-1} + \dots + \M_{p} \X_{t-p } =\V_t , 
  \quad t=p+1, p+2 , \dots ,
\end{equation}
where $\V_t$ is white noise, it
 is uncorrelated with 
$\X_{t-1},\dots ,\X_{t-p}$, it has zero expectation and  
covariance matrix $\Siga$
(not necessarily diagonal,  but positive definite), and the matrices $\M_j$
satisfy the stability conditions (see~\cite{Deistler}). 
(Sometimes $\X_t$ is isolated on the 
left hand side.) $\V_t$ is called \textit{innovation}, 
i.e. the error term of the (added value to the)
best one-step ahead linear prediction of  $\X_{t}$ with its  past,
which (in case of a VAR($p$) model) can be  done with the $p$-lag long past 
$\X_{t-1} , \dots , \X_{t-p}$. 

\item
\textbf{Structural form} SVAR$(p)$ model: for given integer $p\ge 1$, it is
\begin{equation}\label{structural}
 \A\X_{t} + \B_1 \X_{t-1} + \dots + \B_{p} \X_{t-p } = \U_t , \quad
  t=p+1, p+2 , \dots  ,
\end{equation}
where the white noise term $\U_t$ is uncorrelated with 
$\X_{t-1},\dots ,\X_{t-p}$, it has zero expectation with uncorrelated
components, i.e. with positive definite, diagonal
covariance matrix $\Dada$.
$\A$ is $d\times d$ upper triangular matrix with 1s along its main diagonal;
whereas, $\B_1 ,\dots ,\B_p$ are $d\times d$ matrices, 
see also~\cite{Lutkepohl}.

There is a one-to-one correspondence between the reduced and structural
model; since $\A$ is invertible, from Equation~\eqref{structural}, 
Equation~\eqref{reduced} can be obtained (and vice versa):
$$
 \X_{t} + {\A}^{-1} \B_1 \X_{t-1} + \dots +  {\A}^{-1} \B_{p} \X_{t-p } =  
{\A}^{-1} \U_t , \quad  t=p+1, p+2 , \dots  ,
$$
where  $\M_j = {\A}^{-1} \B_j$, $\V_t ={\A}^{-1} \U_t$, and 
$\Siga ={\A}^{-1} \Dada {{\A}^T}^{-1}$; further, $|\Siga |=|\Dada |$
as $| \A | =1$.  

Note that if the ordering of the components of $\X_t$ is changed (with some
permutation of $\{ 1,\dots ,d \}$), then
the rows of the matrices $\M_j$, further, the rows and columns of $\Siga$ 
are permuted accordingly. 
However, the matrices $\A$, $\B_j$ and $\Dada$ cannot be obtained
in this simple way, they profoundly change under the above permutation. 


\item
\textbf{Causal} CVAR$(p)$ \textbf{unrestricted} model: 
it also obeys Equation~\eqref{structural}, but here the ordering
of the components follows a causal ordering, given e.g. by an expert knowledge.
This is a recursive ordering along a ``complete'' DAG, where the permutation
(labeling)
of the graph nodes (assigned to the components of $\X_t$) is such that 
$X_{t,i}$ can be caused by  $X_{t,j}$ whenever $i<j$,  which means a $j\to i$
directed edge. Here the causal effects are meant contemporaneously, and
reflected in the upper triangular structure of the matrix $\A$.

It is important that in any ordering of the jointly Gaussian 
variables, a Bayesian network or a Gaussian directed graphical model
can be constructed, in which every node (variable) is
regressed linearly with the variables corresponding to higher label nodes,
see Appendix~\ref{appB}.
The partial regression coefficients behave like path coefficients, also used
in SEM.
If the DAG is complete, then
there are no zero constraints imposed on the
partial regression coefficients.

\item
\textbf{Causal} CVAR$(p)$ \textbf{restricted} model: 
here an incomplete DAG is built, based on partial correlations, but in
the conditioning set the $p$-lag long past also counts.

First we build an undirected graph: not
connect $i$ and $j$ if the partial correlation coefficient of 
$X_{i}$ and $X_{j}$, eliminating the effect of the other variables 
is 0 (theoretically), or less than a threshold (practically). 
Such an undirected graphical model is called Markov random field (RMF).
It is known (see~\cite{Rao} and~\cite{Lauritzen})
that partial correlations can be calculated from the concentration matrix
(inverse of the covariance matrix). But here the upper left block of the
inverse of the large block matrix, containing the first $p$ autocovariance
matrices, is used.
If this undirected graph is triangulated, 
then in a convenient (so-called perfect)
ordering of the nodes, the zeros of the adjacency matrix form an RZP.
We can find such an ordering (not necessarily unique) with the 
maximal cardinality search (MCS) algorithm, 
together with cliques and separators of a so-called junction tree (JT),
see~\cite{Bollacta} and Appendix~\ref{appB}. 
In this ordering (labeling) of
the nodes, a DAG can also be constructed, which is Markov
equivalent to the undirected one (it has no so-called sink V configuration),
for further details see Subsection~\ref{incomplete}.

Observe that in this DAG, in the Gaussian case,  the fact
that for $i<j$: $X_i$ and $X_j$ are conditionally independent given
$\{ i+1 ,\dots ,d \} \setminus \{ j \}$ is equivalent to the fact that
$X_i$ and $X_j$ are conditionally independent given all the
remaining components.
(It does not mean that the regression or partial correlation coefficients are
the same, just they are zeros at the same time.)

Having an RZP in the CVAR restricted model, we use the incomplete DAG for 
estimation. With the covariance selection method of~\cite{Dempster},
 the starting concentration matrix is  re-estimated by imposing zero
constraints for its entries in the RZP positions (symmetrically).
By the theory (see, e.g.~\cite{Bollacta}), this will result in
zero entries of $\A$ in the no directed edge positions. This also becomes
obvious from Appendix~\ref{appB} and the algorithm of the Appendix~\ref{alg2}. 

\end{itemize}

Note that the unrestricted CVAR model can as well use an incomplete DAG, 
where the
labeling of its nodes follows the perfect labeling of the undirected
graph; still,
the  parameter matrices $\A$ and $\B_j$s are ``full''  in the sense
that no zeros of $\A$ are guaranteed in the no-edge positions of the graph. 
Their entries are just considered as path 
coefficients of the contemporaneous and lagged effects, respectively. 
On the contrary, in the restricted CVAR model, action is done for introducing
zero entries in $\A$ in the no-edge positions. If the desired zeros form
an RZP, the covariance selection has a closed form (see~\cite{Lauritzen}).
In the lack of an RZP, the covariance selection also works, but it needs an
infinite (convergent) iteration, called Iterative Proportional Scaling (IPS),
see Appendix~\ref{appB}.

Then both in the unrestricted and restricted CVAR$(p)$ models an order
selection is initiated to choose the optimal $p$,
based on information criteria, like AIC, BIC, AICC, and 
HQ, where only the number of parameters differs in the two cases. 
Actually, in the restricted case, the covariances are estimated only
within the cliques and separators of the JT, which can reduce the
computational complexity of our algorithm when the number of nodes is
``large''.


\section{Results}\label{results}

\subsection{The Unrestricted Causal VAR(1) Model}\label{model}

The directed Gaussian graphical model of~\cite{Wermuth} (see Appendix~\ref{rec})
does not consider time development; it is, in fact, a CVAR(0) model.
Also note that at this point, the ordering of the jointly 
Gaussian variables is not relevant, since in any recursive ordering of them 
(encoded in $\A$) a Gaussian directed graphical model
(in other words, a Gaussian Bayesian network) can be constructed, where every 
variable is regressed linearly with the higher label ones. 
This is due to the 
solvability of the recursive equation system~\eqref{matrix}  with the LDL 
decomposition~\eqref{LDLdecomp} in any ordering of the rows and columns of
$\Siga$.


To introduce the unrestricted CVAR(1) model, 
let $\{ \X_t \}$ be a $d$-dimensional, weakly  stationary process with
real valued components of zero expectation 
and covariance matrix function $\C (h)$,
$h=0,\pm 1 ,\pm 2 ,\dots$; $\C (-h ) =\C^T (h )$. All deterministic and
random vectors are column vectors and so, $\C (h )=\E \X_t \X_{t+h}^T$
does not depend on $t$, by weak stationarity.
The model equation is
\begin{equation}\label{mimo}
 \A\X_{t} + \B \X_{t-1} = \U_t , \quad t=1,2,\dots  ,
\end{equation}
where  $\A$ is $d\times d$ upper triangular matrix with 1s along its main 
diagonal, $\B$ is a $d\times d$ matrix; further,
the white noise random vector $\U_t$ is uncorrelated with
(in the Gaussian case,
independent of) $\X_{t-1}$, has zero expectation and  
covariance matrix $\Dada =\diag (\delta_1 ,\dots ,\delta_d )$.

Let $\mathfrak{C}_2$ denote the covariance matrix of the stacked random vector
 $(\X^T_{t} , \X^T_{t-1} )^T$ which, in block matrix form, is as follows:
\begin{equation}\label{egyik}
 \mathfrak{C}_2 = \begin{pmatrix}
  \C (0)  & \C^T (1) \\
 \C (1)  & \C(0)
 \end{pmatrix} .
\end{equation}
It is symmetric and positive definite if the process is of full rank regular
(which means that its spectral density matrix of~\cite{Szabados} 
is of full rank)
that is assumed in the sequel. It is well known that the inverse of 
$\mathfrak{C}_2$, the 
so-called concentration matrix $\K$, has the block-matrix form
$$
 \K = \begin{pmatrix}
  \C^{-1} (1|0)  & -\C^{-1} (1|0) \C^T (1) \C^{-1} (0) \\
 -\C^{-1} (0) \C (1) \C^{-1} (1|0) & 
  \C^{-1} (0) +\C^{-1} (0) \C(1) \C^{-1} (1|0) \C^T (1)\C^{-1} (0)
 \end{pmatrix} ,
$$
where $\C (1|0) = \C (0) - \C^T (1) \C^{-1} (0) \C(1)$ is the conditional
covariance matrix $\C (t|t-1)$ of the distribution of $\X_t$ conditioned on 
$\X_{t-1}$; by weak stationarity,
it does not depend on $t$ either, therefore it is denoted by $\C (1|0)$.
Also,  $\mathfrak{C}_2$ is positive definite if and only if
both $\C(0)$ and  $\C (1|0)$ are positive definite.

Observe that $\C (1|0) = \A^{-1} \Dada {\A^{-1}}^T$ is  the covariance 
matrix of the innovation $\V_t ={\A}^{-1} \U_t$. Therefore, the left upper
block of $\K$ contains its inverse, which is  $\A^{T} \Dada^{-1} \A$.

\begin{Theorem}\label{tetel1}
The parameter matrices $\A$, $\B$, and  $\Dada$ of
model Equation~(\ref{mimo}) can be obtained by the block LDL decomposition
of the (positive definite) concentration matrix $\K$  
(inverse of the covariance matrix $\mathfrak{C}_2$ in Equation~\eqref{egyik}) 
 of the $2d$-dimensional Gaussian random vector $(\X^T_{t} , \X^T_{t-1} )^T$.
If $\K =\LL \DD \LL^T$ is this (unique) decomposition  
with block-triangular matrix $\LL$ and block-diagonal matrix 
$\DD$, then they have the form
\begin{equation}\label{fatma}
\LL = \begin{pmatrix}
  \A^T  &\OO_{d\times d} \\

 \B^T  &\I_{d\times d}
 \end{pmatrix} , \qquad
\DD = \begin{pmatrix}
  \Dada^{-1}  &\OO_{d\times d} \\
  \OO_{d\times d}  &\C^{-1} (0) 
 \end{pmatrix} ,
\end{equation}
where  the  $d\times d$ upper triangular matrix $\A$
with 1s along its main 
diagonal, the $d\times d$ matrix $\B$,  
and the diagonal matrix $\Dada$ of
model Equation~(\ref{mimo}) can be retrieved from them.
\end{Theorem}
The proof of this theorem together with the detailed description of the
algorithm is to be found in Sections~\ref{biz1} and~\ref{alg1} of 
Appendix~\ref{appA}.

\subsection{The Restricted Causal VAR(1) Model}\label{incomplete} 

 Assume that we have a causal
ordering of the coordinates $X_1 ,\dots ,X_d$ of $\X$ such that $X_j$ can be
the cause of $X_i$ whenever $i<j$. We can think of $X_i$s as the nodes  of
a graph in a directed graphical model (Bayesian network) and their labeling
corresponds to a topological 
ordering of the nodes of the underlying DAG. 
So $i<j$ can imply an $i\to j$ edge, and then we say that $X_j$ is a parent 
(cause) of $X_i$ ($X_i$ can have multiple parents, maximum $d-i$ ones),
see Section~\ref{bn}.
For example, when asset prices or relative returns of different assets or 
currencies 
(on the same day) influence each other in a certain (recursive) order.
Now, restricted cases are analyzed, when only certain arrows (causes)
are present, but the DAG is connected. 
In particular, only certain asset prices influence some others on a DAG
contemporaneously, but
not all possible directed edges are present.
In this case, a covariance selection technique can be initiated to
re-estimate the covariance matrix so that the partial regression coefficients
in the no-edge positions be zeros.

When the DAG is built from an undirected graph, then
we also require that the so constructed DAG be Markov equivalent
to its undirected skeleton. Then the DAG must not contain sink V 
configuration (see the forthcoming Remark~\ref{sink}), 
which fact is equivalent to having a so-called 
RZP (see~\cite{Wermuth}) in the
adjacency matrix of the undirected graph in the DAG labeling
(also see the forthcoming Definition~\ref{RZPdir}). In this case,
the positions of the zero entries in the concentration matrix  are identical 
to the  positions of the zero entries in $\A$.
More exactly, for $i=1,\dots ,d-1$ and $j=i+1 ,\dots ,d$:
$\beta_{ij \cdot \{ i+1 ,\dots ,d \} }$, the  partial regression 
coefficient of $X_j$ when  regressing $X_i$
with $X_{i+1} ,\dots ,X_d $, is zero exactly when 
$r_{ij \cdot \{ 1 ,\dots ,d \} \setminus \{ i,j \}  }$ (the  partial correlation 
coefficient between $X_i$ and $X_j$, excluding the effect of all other
variables) is zero.
In other wording, if the DAG does not contain sink V configuration, then it
is Markov equivalent to it undirected skeleton, see~\cite{Lauritzen}
and Corollary~\ref{moral} of Appendix~\ref{decomp}.

In the sequel, we will consider these so-called decomposable cases, 
though there are algorithms, like the Iterative Proportional Scaling (IPS)
for the general case too.
About the relation of the directed and undirected cases, and
equivalent notions of decomposability, e.g. existence of a junction tree (JT)
structure, see Appendix~\ref{decomp}.


In the unrestricted model, no restrictions for the upper-diagonal entries of 
$\A$ were made. 
In practice, we have a sample and all the autocovariance matrices are estimated,
and the resulting $\A ,\B$ matrices are calculated with them. Usually a
statistical hypothesis testing advances this procedure, during which it can be
found that certain
partial correlations (closely related to the entries of $\K$) do 
not significantly differ from zero. Then we naturally want to introduce
zeros for the corresponding entries of $\A$. 
For this, the method of covariance 
selection of~\cite{Dempster} is elaborated, 
see also~\cite{Lauritzen}, \cite{Wermuth}, and Appendix~\ref{undirg}, from which
we recall the following notions here.

\begin{Definition}\label{RZPundir}
Let $\M$  be the adjacency matrix of an undirected graph on node set
labelled as $\{ 1,\dots ,d \}$; 
i.e. for $i\ne j$:  $m_{ij}=m_{ji}=1$ if $i\sim j$ ($i$ and $j$ are 
connected) and
0, otherwise (the diagonal entries are zeros). We say that $\M$ has a
\textbf{reducible zero pattern (RZP)} if $m_{ij}=0$ $(i<j)$ implies that
for each $h=1,\dots ,i-1$: either $m_{hi} =0$ or  $m_{hj} =0$ holds (or both
hold). (This applies to the entries above the main diagonal and 
symmetrically extends
to those below the main diagonal: $m_{ij} =m_{ji}$ for $j<i$.)
\end{Definition}

\begin{Definition}\label{RZPdir}
Let $\M$  be the adjacency matrix of a DAG in the topological ordering
$1,\dots ,d$ of its nodes; i.e. for $i<j$:  $m_{ij}=1$ if there is a
$j \to i$ edge and
0, otherwise ($\M$ is upper triangular with zero diagonal). 
We say that $\M$ has a
\textbf{reducible zero pattern (RZP)} if $m_{ij}=0$ $(i<j)$ implies that
for each $h=1,\dots ,i-1$: either $m_{hi} =0$ or  $m_{hj} =0$ holds (or both
hold).
\end{Definition}

\begin{Remark}\label{sink}
Obviously, in the adjacency matrix of a DAG, an
RZP is present if and only if there is no sink V configuration in the 
topological ordering of the DAG.
Under \textbf{sink V} configuration a triplet $j\rightarrow h \leftarrow i$ 
is understood, where $i$ is not connected by $j$ (see Figure~\ref{FATMA} of
Appendix~\ref{bn}). 
Indeed, in this case the DAG has a triplet
$h<i<j$ with $m_{hj} \neq 0$,
$m_{hi} \neq 0$, but $m_{ij} =0$,
in contrast to Definition~\ref{RZPdir}. 
\end{Remark}
Observe that the RZP applies only to the  positions of the zero entries of 
a matrix.  So we are able to give a more general definition.
\begin{Definition}
Let $\M$  be a symmetric or an upper triangular matrix of real entries. 
We say that $\M$ has a
\textbf{reducible zero pattern (RZP)} if $m_{ji}=0$ $(j<i)$ implies that
for each $h=1,\dots ,j-1$: either $m_{hj} =0$ or  $m_{hi} =0$ holds (or both
hold).
\end{Definition}
In view of this, we can find relation between the zeros of $\A$ in the
CVAR(1) model  and those of the inverse covariance matrix.
\begin{Proposition}
The upper triangular matrix  $\A$ of model equation~\eqref{mimo} has 
an RZP  if and only if the upper left
$d\times d$ block of $\K =\mathfrak{C}_2^{-1}$ has an RZP.
Moreover, the zero entries of $\A$ are exactly in the same positions as the
zero entries of the upper diagonal part of the upper left block of $\K$.
\end{Proposition}

\begin{proof}
In one direction, the proof follows from Equation~\eqref{*} 
of Section~\ref{alg1} 
in Appendix~\ref{appA}. 
Indeed, in presence of an RZP in $\A$, for $1\le j<i \le d$:
if $l_{ij} (=a_{ji}) =0$, then $k_{ij}=0$, since either 
$\ell_{ih} (=a_{hi} )=0$ or
$\ell_{jh} (=a_{hj} )=0$ (or both) for $h=1,\dots, j-1$, which are the intrinsic
 entries of the summation in Equation~\eqref{*}.

In the other direction, if $k_{ij}=0$, then $l_{ij}=a_{ji}=0$ too, because of
the Markov equivalence of the DAG and its undirected skeleton in the 
decomposable case; see Corollary~\ref{moral} of Appendix~\ref{decomp}. 
Here we use
that  $k_{ij}$ and $l_{ij}$ are positive constant multiples of partial
correlation and partial regression coefficients.
\end{proof}

Consequently, if we have  causal relations between the contemporaneous
components of $\X_t$, and the so constructed DAG has an RZP, then this RZP 
is inherited by the left upper block of $\K$, which is
$\C^{-1} (1|0)$.
 Therefore, we further improve the 
covariance selection model~\cite{Dempster}, by introducing  zero entries
into the sample conditional covariance matrix. Actually, 
fixing the zero entries in the left upper block 
of $\K$, we re-estimate the matrix $\mathfrak{C_2}$.
In the possession of
a sample, there are exact MLE estimates developed for this purpose
(see~\cite{Bollacta},~\cite{Lauritzen}).


Assume that the cliques of the node set $\{ 1,\dots ,d \}$ of $\X_t$ are
$C_1 , \dots ,C_{k}$, to which a last clique  is added, formed 
by the components $X_{t-1,1}, \dots ,X_{t-1,d}$
of $\X_{t-1}$. If $C_1 , \dots ,C_{k}$ form a JT in this
 ordering, the joint density of $\X_t$ and $\X_{t-1}$ 
factorizes like
$$
 f(\x_t ,\x_{t-1} ) =f (\x_{t-1} ) \prod_{j=1}^{k} 
 f (\x_{t,R_{j}} \, | \, \x_{t,S_{j } }, \x_{t-1}  ) .
$$

For covariance selection, 
we include the lag 1 variables $X_{t-1,1}, \dots ,X_{t-1,d}$ too. Therefore,
the new cliques and separators are
$$
 C'_j := C_j \cup \{ X_{t-1,1}, \dots ,X_{t-1,d} \} , \quad j=1,\dots ,k
$$
and
$$
 S'_j := S_j \cup \{ X_{t-1,1}, \dots ,X_{t-1,d} \} , \quad j=2,\dots ,k .
$$

Having this, we are able to re-estimate the  $2d\times 2d$ $\K$, inverse of
$\mathfrak{C}_2$ in~\eqref{egyik}, for our VAR(1) model as follows:
\begin{equation}\label{khat}
 {\hat \K} = (n-1) \left\{ \sum_{j=1}^k [\SSS_{C'_j}^{-1}]^{2d } -
             \sum_{j=2}^k [\SSS_{S'_j}^{-1}]^{2d }  \right\} ,
\end{equation}
where the matrix $\SSS_{C'}$ is the usual product-moment
estimate based on the $n-1$ element  sample with the
following variables: 
$$
 \left( X_{t,i}: \, i\in C \quad \textrm{and} 
 \quad X_{t-1,1}, \dots ,X_{t-1,d} \right), \quad t=2,\dots ,n ;
$$
further,    $[\M_{C'} ]^{2d}$ denotes the $2d\times 2d$ 
matrix comprising 
the entries of the larger $2d \times 2d$ matrix 
$\M$ in the $|C'| \times |C'|$ block corresponding to $C'$, 
and otherwise zeros. 
By the properties of the 
LDL decomposition, these zeros go into zeros of $\A$.
For illustrations, see Section~\ref{appl} and~\ref{undirg}.

Note that for the MLE, here we do not have an i.i.d. sample, but a serially 
correlated sample. 
However, by ergodicity, for ``large'' $n$, this choice also works and
gives an asymptotic MLE, akin to the
product-moment estimates.

\subsection{Higher Order Causal VAR Models}\label{high}

The above model can be further generalized to the following recursive 
VAR($p$) model: for given integer $p\ge 1$,
\begin{equation}\label{mimop}
 \A\X_{t} + \B_1 \X_{t-1} + \dots + \B_{p} \X_{t-p } = \U_t , 
  \quad t=p+1, p+2 , \dots  ,
\end{equation}
where the white noise term $\U_t$ is uncorrelated with 
$\X_{t-1},\dots ,\X_{t-p}$, it has zero expectation and  
covariance matrix $\Dada =\diag (\delta_1 ,\dots ,\delta_d )$.
$\A$ is $d\times d$ upper triangular matrix with 1s along its main diagonal;
whereas, $\B_1 ,\dots ,\B_p$ are $d\times d$ matrices.
 
Here we have to perform the block Cholesky decomposition of the inverse
covariance matrix of $(\X_{t}^T ,\X_{t-1}^T ,\dots ,\X_{t-p}^T )^T$,
i.e. the inverse  of the matrix
\begin{equation}\label{eq:block_Toeplitz}
\mathfrak{C}_{p+1} = \left(
  \begin{array}{lllll}
    \bm{C}(0)  & \bm{C}^T(1) & \bm{C}^T(2) & \cdots & \bm{C}^T(p) \\
    \bm{C}(1) & \bm{C}(0) & \bm{C}^T(1) & \cdots & \bm{C}^T(p-1) \\
    \bm{C}(2) & \bm{C}(1) & \bm{C}(0) & \cdots & \bm{C}^T(p-2) \\
    \vdots & \vdots & \vdots & \ddots & \vdots \\
    \bm{C}(p) & \bm{C}(p-1) & \bm{C}(p-2) & \cdots & \bm{C}(0) \\
  \end{array}
\right) .
\end{equation}
This is a symmetric, positive definite block Toeplitz matrix with 
$(p+1)\times (p+1)$ blocks  which are $d\times d$ matrices. Again,
$\C^T (h)=\C (-h )$ and 
 it is well known that the inverse matrix 
$\mathfrak{C}_{p+1}^{-1}$ has the following block-matrix form:
\begin{itemize}
\item Upper left block: $\C^{-1} (p|0,\dots ,p-1)$;
\item Upper right block: $-\C^{-1}(p|0,\dots ,p-1)\C^T (1,\dots ,p)
  \mathfrak{C}_p^{-1} $;
\item Lower left block: $-\mathfrak{C}_p^{-1}\C (1,\dots ,p) 
   \C^{-1}(p|0,\dots ,p-1)$;
\item Lower right block: 
 $\mathfrak{C}_p^{-1}  +\mathfrak{C}_p ^{-1}  \C(1,\dots ,p) 
\C^{-1} (p|0,\dots ,p-1)  \C^T(1,\dots ,p) \mathfrak{C}_p^{-1} $, 
\end{itemize} 
where $\C (p|0,\dots ,p-1) = \C (0) - \C^T (1,\dots ,p) \mathfrak{C}_p^{-1} 
\C(1,\dots ,p)$ is the conditional
covariance matrix $\C (t|t-1,\dots ,t-p)$ of the distribution of $\X_t$ 
given  $\X_{t-1} ,\dots ,\X_{t-p}$; due to stationarity,
it does not depend on $t$ either, therefore it is denoted by 
$\C (p|0,\dots ,p-1 )$. Further, $\C^T (1,\dots ,p)=(\C^T (1),\dots ,\C^T (p))$
is $d \times pd$ and $\C (1,\dots ,p)$ is $pd \times d$.
Also,  $\mathfrak{C}_{p+1}$ is positive definite if and only if
both $\mathfrak{C}_p$ and  $\C (p|0,\dots ,p-1)$ are positive definite.

\begin{Theorem}\label{tetel2}
The parameter matrices $\A$, $\B_1 ,\dots ,\B_p$ and  $\Dada$ of
model Equation~(\ref{mimop}) can be obtained by the block LDL decomposition
of the (positive definite) concentration matrix $\K$  
(inverse of the covariance matrix $\mathfrak{C}_{p+1}$ in 
Equation~\eqref{eq:block_Toeplitz}) 
 of the $(p+1)d$-dimensional Gaussian random vector
 $(\X_{t}^T ,\X_{t-1}^T ,\dots ,\X_{t-p}^T )^T$.
If $\K =\LL \DD \LL^T$ is this (unique) decomposition  
with block-triangular matrix $\LL$ and block-diagonal matrix 
$\DD$, then they have the form
\begin{equation}\label{fatma1}
\LL = \begin{pmatrix}
  \A^T  &\OO_{d\times pd} \\
 \B^T  &\I_{pd\times pd}
 \end{pmatrix} , \qquad
\DD = \begin{pmatrix}
  \Dada^{-1}  &\OO_{d\times pd}  \\
  \OO_{pd\times d}  &\mathfrak{C}_{p}^{-1} 
 \end{pmatrix} ,
\end{equation}
where  the  $d\times d$ upper triangular matrix $\A$
with 1s along its main  diagonal, the 
 $d\times pd$ matrix $\B=(\B_1 \dots \B_p )$ (transpose of $\B^T$,
partitioned into blocks)
and the diagonal matrix $\Dada$ of
model Equation~(\ref{mimop}) can be retrieved from them.
\end{Theorem}
The proof of this theorem together with the detailed description of the
algorithm is to be found in Sections~\ref{biz2} and~\ref{alg2} of 
Appendix~\ref{appA}.

Restricted cases can be treated similarly as in Section~\ref{incomplete}.
Here too, the existence of an RZP in the DAG on $d$ nodes is equivalent
to the  existence of an RZP in the left upper $d\times d$ corner of the
concentration matrix $\mathfrak{C}_{p+1}^{-1}$. From the model equations
it is obvious that
$$
 X_{t,i} = -\sum_{j=i+1}^d a_{ij} X_{t,j} - \sum_{h=1}^p \sum_{j=1}^d 
  b_{h,ij} X_{t-h, j} -U_{t,i} , 
$$
where $X_{t,i}$ is the $i$th coordinate of $\X_t$.
By weak stationarity it follows that the entries of the matrices
$\A$ and $\B_h =(b_{h,ij} )_{i,j=1}^{p}$ are partial regression coefficients
as follows:
$$
\begin{aligned}
 a_{ij} &= -\beta_{X_{t,i} X_{t,j} \cdot \{ X_{t,i+1} ,\dots ,X_{t,d} ,
          X_{t-1,1} ,\dots ,X_{t-1,d} , \dots , X_{t-p,1} ,\dots ,X_{t-p,d}\} }, 
 \quad 1 \le i <j \le d; \\
  b_{h,ij} &= -\beta_{X_{t,i} X_{t+h, j} \cdot \{ X_{t,i+1} ,\dots ,X_{t,d} ,
     X_{t-1,1} ,\dots ,X_{t-1,d} , \dots ,
    X_{t-p,1} ,\dots ,X_{t-p,d}  \} }, \\
   &1 \le i <j \le d, \quad h=1,\dots ,d.
\end{aligned}
$$
Since the conditioning set changes from equation to equation, it is
easier to use the block LDL decompositions here, without the exact meaning
of the coefficients.

Considering the components of $\X_t ,\X_{t-1} ,\dots ,\X_{t-p}$ as 
nodes of the expanded graph,
the joint density of $\X_t ,\X_{t-1} ,\dots ,\X_{t-p}$ factorizes like
$$
\begin{aligned}
 f(\x_t ,\x_{t-1} , \dots ,\x_{t-p} ) &= f(\x_{t-1} , \dots ,\x_{t-p} )
             f(\x_t \, | \, \x_{t-1} , \dots ,\x_{t-p} ) \\
&=f(\x_{t-1} , \dots ,\x_{t-p} ) \cdot
 \prod_{i=1}^d f(x_{t,i} \, | \, \x_{t,\pa (i) }, \x_{t-1}, \dots ,\x_{t-p} ) .
\end{aligned}
$$

Now assume that the cliques of the node set $\{ 1,\dots ,d \}$ of $\X_t$ are
$C_1 , \dots ,C_{k }$ (as at the end of Section~\ref{incomplete})
and they form a
JT with residuals $R_1 , \dots , R_{k}$ and separators
$S_1 , \dots , S_{k }$ (with the understanding that $S_1 =\emptyset$ and 
$R_1 =C_1$).
Enhancing the preceding density with this, we get the following factorization:
$$
 f(\x_t ,\x_{t-1} , \dots ,\x_{t-p} ) =f(\x_{t-1} , \dots ,\x_{t-p} )  \cdot
 \prod_{j=1}^{k } 
 f (\x_{t, R_{j}} \, | \, \x_{t, S_{j } }, \x_{t-1} , \dots , \x_{t-p} ) .
$$

Covariance selection can be done  similarly as in Section~\ref{incomplete},
but here zero entries of the left upper $d\times d$ block of 
$\mathfrak{C}_{p+1}^{-1}$ provide the zero  entries of $\A$.
For this purpose, the
 $n-p$ element  sample  entries are used with the following coordinates: 
$$
 \left( X_{t,i}: \, i\in C \quad \textrm{and} \quad X_{t-1,1}, \dots ,
X_{t-1,d},   \dots , X_{t-p,1}, \dots ,X_{t-p,d} \right),
$$
for $t=p+1,\dots ,n$ 
when we calculate the product-moment estimate $\SSS_{C'}$ with
$C' =C \cup \{  \X_{t-1} ,\dots ,\X_{t-p} \}$.
For more details see the explanation after Equation~\eqref{khat} and 
Section~\ref{appl}. 

Note that here the covariance selection is done based only on a serially
correlated and not an independent sample. However, when $n$ is ``large'', then
ergodicity issues (see, e.g.~\cite{Szabados}) give rise to this 
relaxation of the original algorithm.
Also, by the theory of~\cite{Brockwell} (p. 424), 
it is guaranteed that the Yule-Walker equations have a 
stable
stationary solution to the VAR$(p)$ model, whenever the starting  covariance
matrix of $(p+1)\times (p+1)$ blocks is positive definite. But we assume this
in our theorems. In this case, the empirical versions are also positive
definite (almost surely as $n\to\infty$), and the covariance selection also
gives a positive definite estimate. So the estimated parameter matrices 
provide a stable VAR model in view of the theory and ergodicity 
if $n$ is ``large''.

\section{Applications with Order Selection}\label{appl}


\subsection{Financial Data} 

We used the data communicated in the  the paper~\cite{Akbilgic} on
daily relative returns of 8 different asset prices, spanning 534 days.
The multivariate time series was found stationary and nearly Gaussian.

First we applied the unrestricted CVAR($p$) model.
We constructed a DAG by making the undirected graph on 8 nodes
directed. The undirected graph was constructed by testing statistical
hypotheses for the partial correlations of the pairs of the variables
conditioned on all the others. As the test statistic is increasing in the
absolute value of the partial correlation in question, a threshold 0.04
for the latter one was used that corresponds to significance level 
$\alpha =0.008851$ of the partial correlation test.
Table~\ref{tab:part_cor_C(0)}   contains  the
partial correlations based on $\C^{-1} (0)$.

\begin{table}[ht]
	\centering
	\begin{tabular}{r c c c c c c c c}
		&\textbf{NIK} & \textbf{EU} & \textbf{ISE} & \textbf{EM} & \textbf{BVSP} & \textbf{DAX} & \textbf{FTSE} & \textbf{SP} \\ 
		\textbf{NIK} &  & 0.016$^{*}$ & 0.035$^{*}$ & 0.522 & -0.260 & -0.019$^{*}$ & -0.076 & 0.024$^{*}$ \\ 
		\textbf{EU} & 0.016$^{*}$ &	& 0.217 & 0.034$^{*}$ & 0.067 & 0.687 & 0.747 &	0.018$^{*}$ \\ 
		\textbf{ISE} & 0.035$^{*}$ & 0.217 & & 0.358 & -0.157 &	-0.077 & -0.059 & 0.034$^{*}$ \\ 
		\textbf{EM} & 0.522 & 0.034$^{*}$ & 0.358 & & 0.546 & 0.048 & 0.086 & -0.184 \\ 
		\textbf{BVSP} & -0.260 & 0.067 & -0.157 & 0.546 & &	-0.093 & -0.045 & 0.533 \\ 
		\textbf{DAX} &
	-0.019$^{*}$ &
		0.687 &
		-0.077 &
		0.048 &
		-0.093 &
		&
		-0.203 &
		0.191 \\ 
		\textbf{FTSE} &
		-0.076 &
		0.747 &
		-0.059 &
		0.086 &
		-0.045 &
		-0.203 &
		&
		0.057 \\ 
		\textbf{SP} &
	0.024$^{*}$ &
	0.018$^{*}$ &
	0.034$^{*}$ &
		-0.184 &
		0.533 &
		0.191 &
		0.057 & 		\\ 
	\end{tabular}
	\caption{Partial Correlation Coefficients from $\C^{-1} (0)$.
Entries marked by asterisk are less than 0.04 in absolute value (i.e.
 they correspond to no-edge positions in the graph), the corresponding
significance is $\alpha =0.008851$.}
	\label{tab:part_cor_C(0)}
\end{table}

Since the graph was triangulated, 
with the MCS algorithm, we were able
to (not necessarily uniquely) label the nodes so that the adjacency
matrix of this undirected graph had an RZP: 
$$
\begin{aligned}
 &1: \, \textrm{NIK (stock market return index of Japan)},  \\
 &2: \, \textrm{EU (MSCI European index)}, \\
 &3: \, \textrm{ISE (Istanbul stock exchange national 100 index)} \\
 &4: \, \textrm{EM (MSCI emerging markets index)}, \\
 &5: \, \textrm{BVSP (stock market return index of Brazil)}, \\
 &6: \, \textrm{DAX (stock market return index of Germany)},  \\
 &7: \, \textrm{FTSE (stock market return index of UK)},  \\
 &8: \, \textrm{SP (Standard \& poor's 500 return index)}.
\end{aligned}
$$

\begin{figure}
	\centering
	\subfloat[MRF fitted to the financial dataset.]{\label{fig:undir}
		\includegraphics[width=.48\linewidth]{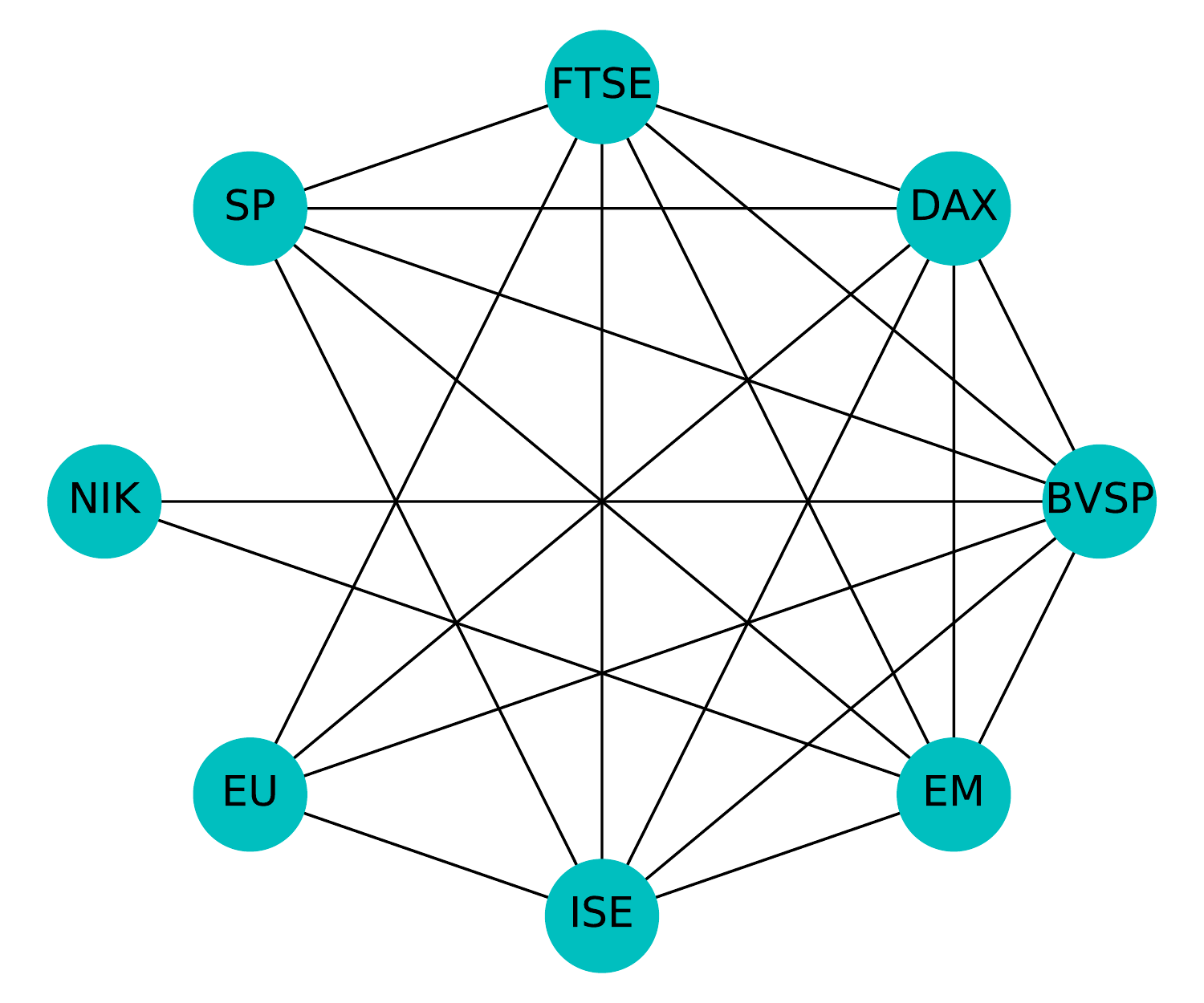}}
	\hfill
	\subfloat[DAG oriented in the MCS ordering.]{\label{fig:dir}
		\includegraphics[width=.48\linewidth]{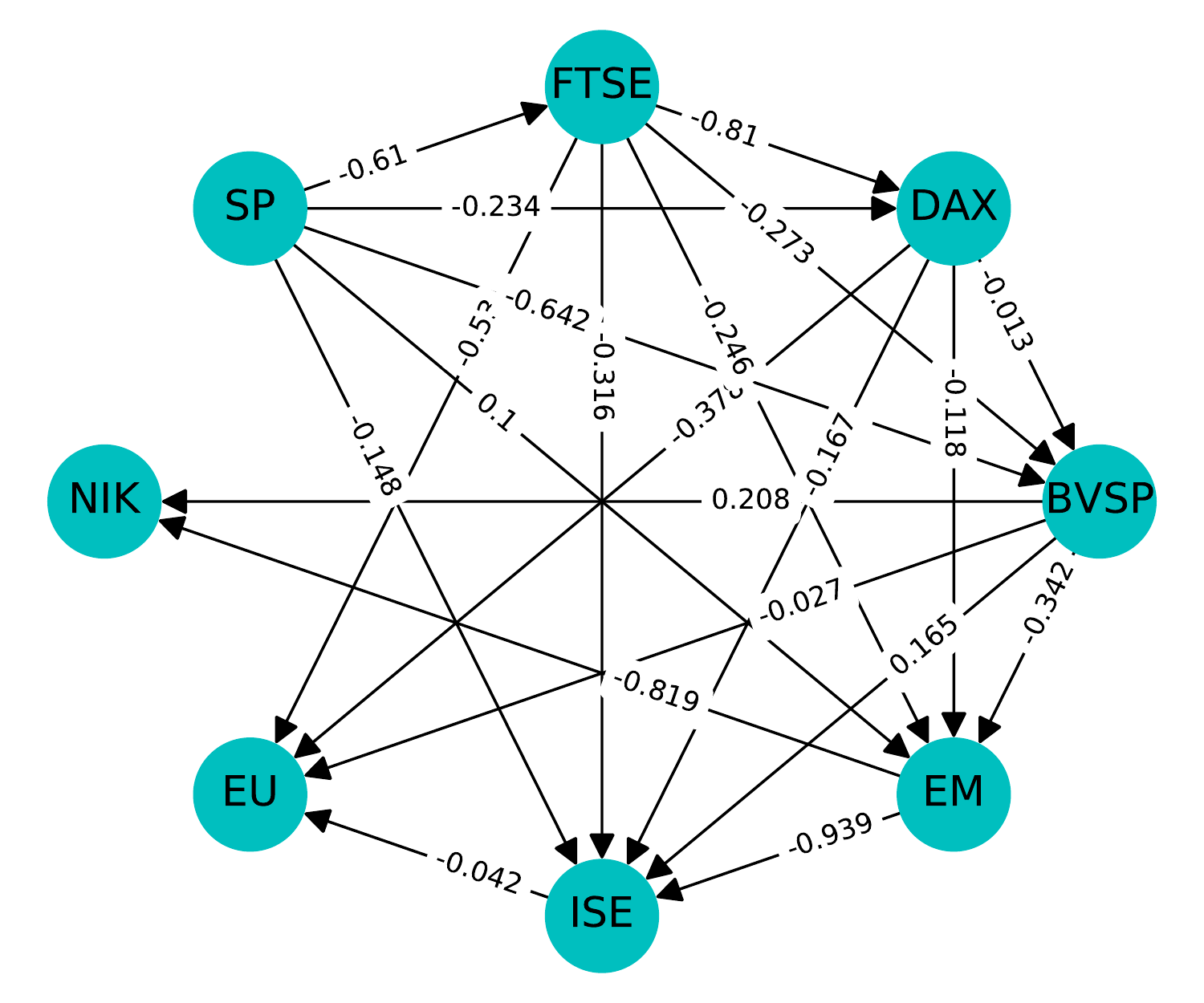}}
\caption{Graphical models fitted to the financial dataset}
\end{figure}

If this is considered as the topological labeling of the DAG, where directed
edges point from a higher label node to a lower label one, then the
so obtained directed graph is Markov equivalent to its undirected
skeleton, see Figures~\ref{fig:undir} and~\ref{fig:dir}. 
But the RZP is used only in the restricted case, in the unrestricted case
only the DAG ordering of the variables is used.


We ran the VAR($p$) algorithm with $p=1,2,3,4,5$ and
found that the $\A$ matrices do not change much with increasing $p$, akin
to $\B_1$. The $\B_2 ,..,\B_5$ matrices have relatively ``small'' entries.
Consequently, contemporaneous effects and one-day lags are the most important.
This is also supported by the forthcoming order selection investigations.
For the $p=1$ and $p=2$ cases, 
see Tables~\ref{tab:A_Fin_VAR1},\ref{tab:B_Fin_VAR1} 
and~\ref{tab:A_Fin_VAR2},\ref{tab:B1_Fin_VAR2},\ref{tab:B2_Fin_VAR2},
respectively. 

\begin{table}[ht]
	\centering
	\scriptsize
	\begin{tabular}{r c c c c c c c c }
		\textbf{} & \textbf{NIK} & \textbf{EU} & \textbf{ISE} & \textbf{EM} & \textbf{BVSP} & \textbf{DAX} & \textbf{FTSE} & \textbf{SP} \\ 
		\textbf{NIK}  & 1 & 0.0264 & 0.0042  & -0.8902 & 0.2030  & 0.0170  & 0.0781  & -0.0336 \\ 
		\textbf{EU}   & 0 & 1      & -0.0418 & -0.0146 & -0.0239 & -0.3746 & -0.5255 & -0.0033 \\ 
		\textbf{ISE}  & 0 & 0      & 1       & -0.9518 & 0.1613  & -0.1658 & -0.3129 & -0.1413 \\ 
		\textbf{EM}   & 0 & 0      & 0       & 1       & -0.3507 & -0.1182 & -0.2464 & 0.1077  \\ 
		\textbf{BVSP} & 0 & 0      & 0       & 0       & 1       & -0.0129 & -0.2782 & -0.6375 \\ 
		\textbf{DAX}  & 0 & 0      & 0       & 0       & 0       & 1       & -0.8102 & -0.2336 \\ 
		\textbf{FTSE} & 0 & 0      & 0       & 0       & 0       & 0       & 1       & -0.6100 \\ 
		\textbf{SP}   & 0 & 0      & 0       & 0       & 0       & 0       & 0       & 1       \\ 
	\end{tabular}
	\caption{$\bm{A}$ for the unrestricted Financial VAR(1) model (rounded to 4 decimals).}
	\label{tab:A_Fin_VAR1}
\end{table}

\begin{table}[ht]
	\centering
	\scriptsize
	\begin{tabular}{r c c c c c c c c}
		\textbf{} &
		\textbf{NIK$_{-1}$} &
		\textbf{EU$_{-1}$} &
		\textbf{ISE$_{-1}$} &
		\textbf{EM$_{-1}$} &
		\textbf{BVSP$_{-1}$} &
		\textbf{DAX$_{-1}$} &
		\textbf{FTSE$_{-1}$} &
		\textbf{SP$_{-1}$} \\ 
		\textbf{NIK}  & 0.1845  & -0.1685 & -0.0874 & 0.0852  & 0.0635  & 0.0205  & -0.1236 & -0.2798 \\ 
		\textbf{EU}   & -0.0131 & 0.1219  & -0.0044 & 0.0291  & -0.0124 & -0.0393 & -0.0979 & 0.0011  \\ 
		\textbf{ISE}  & 0.0677  & 0.2811  & -0.0657 & 0.2473  & -0.2940 & -0.0543 & 0.0098  & -0.1442 \\ 
		\textbf{EM}   & -0.0016 & -0.0569 & -0.0159 & 0.1076  & -0.0917 & -0.0945 & 0.0875  & -0.1071 \\ 
		\textbf{BVSP} & -0.0140 & 0.0704  & 0.0142  & -0.1046 & 0.1397  & -0.1497 & 0.1188  & -0.0812 \\ 
		\textbf{DAX}  & -0.0034 & 0.2021  & -0.0342 & -0.0044 & -0.0352 & -0.0476 & -0.0670 & -0.0673 \\ 
		\textbf{FTSE} & 0.0293  & -0.0168 & -0.0109 & 0.0420  & -0.1129 & 0.2141  & 0.0805  & -0.2641 \\ 
		\textbf{SP}   & 0.0417  & 0.2603  & -0.0261 & 0.0112  & -0.0026 & -0.0709 & -0.2850 & 0.1240  \\ 
	\end{tabular}
\caption{$\bm{B}$ for the unrestricted Financial VAR(1) model 
(rounded to 4 decimals).}
	\label{tab:B_Fin_VAR1}
\end{table}

\begin{table}[ht]
	\centering
	\scriptsize
	\begin{tabular}{r c c c c c c c c}
		\textbf{} & \textbf{NIK} & \textbf{EU} & \textbf{ISE} & \textbf{EM} & \textbf{BVSP} & \textbf{DAX} & \textbf{FTSE} & \textbf{SP} \\ 
		\textbf{NIK}  & 1 & -0.0114 & 0.0103  & -0.8822 & 0.1995  & 0.0233  & 0.0856  & -0.0214 \\ 
		\textbf{EU}   & 0 & 1       & -0.0426 & -0.0110 & -0.0240 & -0.3745 & -0.5137 & -0.0128 \\ 
		\textbf{ISE}  & 0 & 0       & 1       & -0.9788 & 0.1701  & -0.1669 & -0.3139 & -0.1361 \\ 
		\textbf{EM}   & 0 & 0       & 0       & 1       & -0.3450 & -0.1154 & -0.2375 & 0.0922  \\ 
		\textbf{BVSP} & 0 & 0       & 0       & 0       & 1       & -0.0047 & -0.2655 & -0.6601 \\ 
		\textbf{DAX}  & 0 & 0       & 0       & 0       & 0       & 1       & -0.8120 & -0.2339 \\ 
		\textbf{FTSE} & 0 & 0       & 0       & 0       & 0       & 0       & 1       & -0.6320 \\ 
		\textbf{SP}   & 0 & 0       & 0       & 0       & 0       & 0       & 0       & 1       \\ 
	\end{tabular}
\caption{$\bm{A}$ for the unrestricted Financial VAR(2) model 
(rounded to 4 decimals).}
	\label{tab:A_Fin_VAR2}
\end{table}

\begin{table}[ht]
	\centering
	\scriptsize
	\begin{tabular}{r c c c c c c c c}
		\textbf{} &
		\textbf{NIK$_{-1}$} &
		\textbf{EU$_{-1}$} &
		\textbf{ISE$_{-1}$} &
		\textbf{EM$_{-1}$} &
		\textbf{BVSP$_{-1}$} &
		\textbf{DAX$_{-1}$} &
		\textbf{FTSE$_{-1}$} &
		\textbf{SP$_{-1}$} \\ 
		\textbf{NIK}  & 0.2063  & -0.1826 & -0.1106 & 0.1063  & 0.0731  & 0.0187  & -0.1502 & -0.2580 \\ 
		\textbf{EU}   & -0.0037 & 0.1364  & -0.0010 & 0.0232  & -0.0150 & -0.0371 & -0.0996 & -0.0107 \\ 
		\textbf{ISE}  & 0.0409  & 0.2476  & -0.0771 & 0.2274  & -0.2772 & -0.0447 & 0.0331  & -0.1284 \\ 
		\textbf{EM}   & 0.0489  & -0.0200 & -0.0030 & 0.1360  & -0.1150 & -0.0996 & 0.0468  & -0.1162 \\ 
		\textbf{BVSP} & -0.0066 & 0.0931  & 0.0261  & -0.1091 & 0.1312  & -0.1573 & 0.1161  & -0.0935 \\ 
		\textbf{DAX}  & -0.0123 & 0.2146  & -0.0319 & 0.0073  & -0.0406 & -0.0536 & -0.0727 & -0.0694 \\ 
		\textbf{FTSE} & 0.0852  & 0.0019  & 0.0275  & 0.0145  & -0.1117 & 0.2377  & 0.1035  & -0.3427 \\ 
		\textbf{SP}   & 0.0530  & 0.2759  & -0.0565 & -0.0033 & 0.0024  & -0.0945 & -0.3106 & 0.1789  \\ 
	\end{tabular}
\caption{$\bm{B_1}$ for the unrestricted Financial VAR(2) model 
(rounded to 4 decimals).}
	\label{tab:B1_Fin_VAR2}
\end{table}

\begin{table}[ht]
	\centering
	\scriptsize
	\begin{tabular}{ r c c c c c c c c c }
		\textbf{} &
		\textbf{NIK$_{-2}$} &
		\textbf{EU$_{-2}$} &
		\textbf{ISE$_{-2}$} &
		\textbf{EM$_{-2}$} &
		\textbf{BVSP$_{-2}$} &
		\textbf{DAX$_{-2}$} &
		\textbf{FTSE$_{-2}$} &
		\textbf{SP$_{-2}$} \\ 
		\textbf{NIK}  & -0.0402 & -0.1695 & -0.0410 & 0.0156  & 0.0998  & -0.0406 & 0.1367  & -0.0091 \\ 
		\textbf{EU}   & 0.0017  & 0.0771  & -0.0065 & 0.0054  & 0.0037  & 0.0192  & -0.0762 & -0.0394 \\ 
		\textbf{ISE}  & -0.0142 & -0.1725 & -0.0276 & -0.0088 & 0.0389  & 0.1167  & 0.0826  & 0.0357  \\ 
		\textbf{EM}   & -0.0054 & 0.0650  & -0.0322 & 0.1155  & -0.0695 & -0.0959 & -0.0162 & -0.0270 \\ 
		\textbf{BVSP} & -0.0423 & 0.0332  & -0.0449 & 0.2878  & -0.0717 & -0.0221 & -0.0381 & -0.0120 \\ 
		\textbf{DAX}  & -0.0372 & 0.0177  & 0.0130  & 0.0658  & -0.0360 & -0.0108 & -0.0202 & 0.0059  \\ 
		\textbf{FTSE} & 0.0491  & 0.3107  & -0.0820 & 0.0693  & 0.0299  & 0.0153  & -0.0840 & -0.3038 \\ 
		\textbf{SP}   & 0.0447  & -0.0628 & 0.0804  & -0.1824 & 0.0785  & 0.0133  & -0.1775 & 0.1284  \\ 
	\end{tabular}
\caption{$\bm{B_2}$ for the unrestricted Financial VAR(2) model 
(rounded to 4 decimals).}
	\label{tab:B2_Fin_VAR2}
\end{table}

Then we considered the restricted CVAR(1) model.
Here we want to introduce structural zeros into the matrix $\A$. 
Now the matrix $C^{-1}(1|0)$, the left upper $8\times 8$ corner 
of $\mathfrak{C}_2^{-1}$ is used for covariance selection.
Figure~\ref{fig:dir} 
shows this DAG with the significant path coefficients  above the arrows,
based on Table~\ref{tab:A_rst_Fin_VAR1}.

The ordering of the variables is the same as in the unrestricted case,
but the RZP is a bit different.
The decomposable structure has the following cliques and separators:
\begin{equation}\label{klikkek}
\begin{aligned}
C_1 &= \{ \textrm{BVSP,DAX,EM,FTSE,ISE,SP} \} =\{ 3,4,5,6,7,8 \} \\
C_2 &= \{ \textrm{BVSP,DAX,EU,FTSE,ISE} \}=\{ 2,3,5,6,7 \} \\
C_3 &= \{ \textrm{BVSP,EM,NIK} \} =\{ 1,4,5 \}  \\
S_2 &=  \{ \textrm{BVSP,DAX,FTSE,ISE} \} =\{ 3,5,6,7 \} \\
S_3 &= \{ \textrm{BVSP,EM} \} =\{ 4,5 \}   ,
\end{aligned}
\end{equation}
where the parent clique of both $C_2$ and $C_3$ is $C_1$.
Note, that  the  set of nodes in the second braces is the same, but they 
follow increasing labels so that better see the JT structure. 
variables too; 

The matrices $\A$ and $\B$ were estimated via the algorithm for the LDL
decomposition of $\hat \K$. Here the zeros of the left upper $8\times 8$ block
of $\hat \K$ will necessarily result in the zeros of $\A$ in the same positions.
The upper-diagonal entries of $\A$ and the entries of $\B$ are
considered as path coefficients which represent the contemporaneous and 1-day
lagged effect of the assets to the others, respectively; see 
Table~\ref{tab:A_rst_Fin_VAR1} and Table~\ref{tab:B_rst_Fin_VAR1}.

\begin{table}[ht]
	\scriptsize
	\centering
	\begin{tabular}{r c c c c c c c c}
		&
		\textbf{NIK} &
		\textbf{EU} &
		\textbf{ISE} &
		\textbf{EM} &
		\textbf{BVSP} &
		\textbf{DAX} &
		\textbf{FTSE} &
		\textbf{SP} \\ 
		\textbf{NIK} &
		1 &
	0&
	0&
		-0.8193 &
		0.2080 &
	0&
	0&
	0 \\ 
		\textbf{EU} &
		0 &
		1 &
		-0.0421 &
		0&
		-0.0269 &
		-0.3782 &
		-0.5297 &
		0 \\ 
		\textbf{ISE}  & 0 & 0 & 1 & -0.9386 & 0.1653  & -0.1675 & -0.3161 & -0.1477 \\ 
		\textbf{EM}   & 0 & 0 & 0 & 1       & -0.3419 & -0.1184 & -0.2464 & 0.0997  \\ 
		\textbf{BVSP} & 0 & 0 & 0 & 0       & 1       & -0.0130 & -0.2729 & -0.6423 \\ 
		\textbf{DAX}  & 0 & 0 & 0 & 0       & 0       & 1       & -0.8102 & -0.2336 \\ 
		\textbf{FTSE} & 0 & 0 & 0 & 0       & 0       & 0       & 1       & -0.6104 \\ 
		\textbf{SP}   & 0 & 0 & 0 & 0       & 0       & 0       & 0       & 1       \\ 
	\end{tabular}
	\caption{
		$\bm{A}$ matrix for the restricted Financial VAR(1) model (rounded to 4 decimal places).
	}
	\label{tab:A_rst_Fin_VAR1}
\end{table}

\begin{table}[ht]
	\scriptsize
	\centering
	\begin{tabular}{r c c c c c c c c}
		\textbf{} & \textbf{NIK$_{-1}$} & \textbf{EU$_{-1}$} & \textbf{ISE$_{-1}$} & \textbf{EM$_{-1}$} & \textbf{BVSP$_{-1}$} & \textbf{DAX$_{-1}$} & \textbf{FTSE$_{-1}$} & \textbf{SP$_{-1}$} \\ 
		\textbf{NIK}  & 0.1811  & -0.1797 & -0.0856 & 0.0842  & 0.0739  & -0.0058 & -0.1146 & -0.2662 \\ 
		\textbf{EU}   & -0.0131 & 0.1213  & -0.0046 & 0.0304  & -0.0130 & -0.0415 & -0.0969 & 0.0002  \\ 
		\textbf{ISE}  & 0.0676  & 0.2814  & -0.0658 & 0.2483  & -0.2941 & -0.0567 & 0.0120  & -0.1472 \\ 
		\textbf{EM}   & -0.0016 & -0.0567 & -0.0158 & 0.1067  & -0.0908 & -0.0951 & 0.0890  & -0.1085 \\ 
		\textbf{BVSP} & -0.0139 & 0.0704  & 0.0142  & -0.1041 & 0.1391  & -0.1488 & 0.1195  & -0.0828 \\ 
		\textbf{DAX}  & -0.0034 & 0.2019  & -0.0342 & -0.0046 & -0.0353 & -0.0474 & -0.0669 & -0.0672 \\ 
		\textbf{FTSE} & 0.0292  & -0.0171 & -0.0109 & 0.0419  & -0.1130 & 0.2142  & 0.0807  & -0.2642 \\ 
		\textbf{SP}   & 0.0417  & 0.2608  & -0.0261 & 0.0115  & -0.0026 & -0.0713 & -0.2853 & 0.1239  \\ 
	\end{tabular}
	\caption{
		$\bm{B}$ matrix for the restricted Financial VAR(1) model (rounded to 4 decimal places).
	}
	\label{tab:B_rst_Fin_VAR1}
\end{table}

In the VAR(2) situation, the graph, constructed by $C^{-1}(2|1,0)$ is the same,
 has the same
JT with 3 cliques and the same RZP as based on $C^{-1}(1|0)$. 
It is in accord with our former observation that the
lag 2 or more days effects of the assets to the others is negligible
compared to the 1-day lag effect (the forthcoming order selection 
also supports this). 

Here the $24\times 24$ matrix $\hat \K$ was estimated by adapting 
Equation~\eqref{khat} to the $3d\times 3d$ situation,
by using both the lag 1 and lag 2 variables for covariance selection.
The estimated $\A,\B_1$, and $\B_2$ matrices are shown in 
Tables~\ref{tab:A_rst_Fin_VAR2}, \ref{tab:B1_rst_Fin_VAR2},
and~\ref{tab:B2_rst_Fin_VAR2}.

\begin{table}[ht]
	\centering
	\scriptsize
	\begin{tabular}{r cccccccc}
		\textbf{} &
		\textbf{NIK} &
		\textbf{EU} &
		\textbf{ISE} &
		\textbf{EM} &
		\textbf{BVSP} &
		\textbf{DAX} &
		\textbf{FTSE} &
		\textbf{SP} \\ 
		\textbf{NIK} &
		1 &
	0&
	0&
		-0.8191 &
		0.2076 &
0&
0&
0 \\
		\textbf{EU} &
		0 &
		1 &
		-0.0423 &
	0&
		-0.0293 &
		-0.3811 &
		-0.5192 &
	0 \\
		\textbf{ISE}  & 0 & 0 & 1 & -0.9662 & 0.1790  & -0.1713 & -0.3112 & -0.1470 \\ 
		\textbf{EM}   & 0 & 0 & 0 & 1       & -0.3361 & -0.1153 & -0.2372 & 0.0835  \\ 
		\textbf{BVSP} & 0 & 0 & 0 & 0       & 1       & -0.0069 & -0.2544 & -0.6664 \\ 
		\textbf{DAX}  & 0 & 0 & 0 & 0       & 0       & 1       & -0.8128 & -0.2336 \\ 
		\textbf{FTSE} & 0 & 0 & 0 & 0       & 0       & 0       & 1       & -0.6319 \\ 
		\textbf{SP}   & 0 & 0 & 0 & 0       & 0       & 0       & 0       & 1       \\ 
	\end{tabular}
	\caption{$\bm{A}$ matrix for the restricted Financial VAR(2) model 
(rounded to 4 decimals).}
	\label{tab:A_rst_Fin_VAR2}
\end{table}

\begin{table}[ht]
	\centering
	\scriptsize
	\begin{tabular}{rcccccccc}
		\textbf{} & \textbf{NIK$_{-1}$} & \textbf{EU$_{-1}$} & \textbf{ISE$_{-1}$} & \textbf{EM$_{-1}$} & \textbf{BVSP$_{-1}$} & \textbf{DAX$_{-1}$} & \textbf{FTSE$_{-1}$} & \textbf{SP$_{-1}$} \\ 
		\textbf{NIK}  & 0.2009  & -0.1869 & -0.1098 & 0.1089  & 0.0824  & -0.0079 & -0.1493 & -0.2428 \\ 
		\textbf{EU}   & -0.0038 & 0.1387  & -0.0013 & 0.0260  & -0.0153 & -0.0410 & -0.1027 & -0.0086 \\ 
		\textbf{ISE}  & 0.0353  & 0.2865  & -0.0750 & 0.2479  & -0.2741 & -0.0639 & 0.0101  & -0.1418 \\ 
		\textbf{EM}   & 0.0494  & -0.0218 & -0.0027 & 0.1338  & -0.1144 & -0.0990 & 0.0500  & -0.1177 \\ 
		\textbf{BVSP} & -0.0107 & 0.1202  & 0.0276  & -0.0947 & 0.1327  & -0.1674 & 0.0987  & -0.1030 \\ 
		\textbf{DAX}  & -0.0110 & 0.2072  & -0.0322 & 0.0034  & -0.0412 & -0.0503 & -0.0677 & -0.0675 \\ 
		\textbf{FTSE} & 0.0824  & 0.0176  & 0.0281  & 0.0224  & -0.1104 & 0.2309  & 0.0928  & -0.3463 \\ 
		\textbf{SP}   & 0.0506  & 0.2898  & -0.0560 & 0.0040  & 0.0037  & -0.1010 & -0.3199 & 0.1760  \\ 
	\end{tabular}
	\caption{$\bm{B_1}$ matrix for 
	the restricted Financial VAR(2) model 
(rounded to 4 decimals).}
	\label{tab:B1_rst_Fin_VAR2}
\end{table}

\begin{table}[ht]
	\scriptsize
	\centering
	\begin{tabular}{rcccccccc}
		\textbf{} & \textbf{NIK$_{-2}$} & \textbf{EU$_{-2}$} & \textbf{ISE$_{-2}$} & \textbf{EM$_{-2}$} & \textbf{BVSP$_{-2}$} & \textbf{DAX$_{-2}$} & \textbf{FTSE$_{-2}$} & \textbf{SP$_{-2}$} \\ 
		\textbf{NIK}  & -0.0455 & -0.1847 & -0.0391 & 0.0264  & 0.0906  & -0.0486 & 0.1427  & 0.0089  \\ 
		\textbf{EU}   & 0.0017  & 0.0755  & -0.0058 & 0.0047  & 0.0033  & 0.0179  & -0.0765 & -0.0370 \\ 
		\textbf{ISE}  & -0.0161 & -0.1634 & -0.0290 & -0.0021 & 0.0352  & 0.1113  & 0.0821  & 0.0313  \\ 
		\textbf{EM}   & -0.0056 & 0.0659  & -0.0330 & 0.1189  & -0.0701 & -0.0959 & -0.0167 & -0.0283 \\ 
		\textbf{BVSP} & -0.0430 & 0.0415  & -0.0456 & 0.2906  & -0.0729 & -0.0258 & -0.0389 & -0.0168 \\ 
		\textbf{DAX}  & -0.0369 & 0.0163  & 0.0130  & 0.0656  & -0.0356 & -0.0100 & -0.0203 & 0.0064  \\ 
		\textbf{FTSE} & 0.0485  & 0.3142  & -0.0820 & 0.0716  & 0.0290  & 0.0128  & -0.0845 & -0.3054 \\ 
		\textbf{SP}   & 0.0442  & -0.0606 & 0.0805  & -0.1825 & 0.0778  & 0.0117  & -0.1773 & 0.1281  \\ 
	\end{tabular}
	\caption{$\bm{B_2}$ matrix for the restricted Financial VAR(2) model 
(rounded to 4 decimals).}
	\label{tab:B2_rst_Fin_VAR2}
\end{table}

To find the optimal order $p$, information criteria are suggested, see 
e.g.~\cite{Box,Brockwell}. Here the following criteria will be used:
the AIC (Akaike Information Criterion), the AICC (bias corrected version of the AIC), the BIC (Bayesian information criterion), and the HQ (Hannan and Quinn's criterion).
Each criterion can be decomposed into two terms: an \textit{information term} that quantifies the information brought by the model (via the likelihood) 
and a \textit{penalization term} that penalizes too ``large'' number of parameters, to avoid over-fitting.
 It can be proven that the AIC has a positive probability of
overspecification and the BIC is strongly consistent, but sometimes it
underspecifies the true model.
The explicit forms of AIC, BIC, and HQ, which are to be minimized with 
respect to $p$, are as follows: 
$$
\begin{aligned}
\textrm{AIC}(p) &= \ln |{\hat \Dada} | + 
\frac{2 \left[ pd^2 +\binom{d}{2}  \right]}{n-p} =
\sum_{j=1}^d \ln {\hat \delta}_j  + 
\frac{2 \left[ pd^2 +\binom{d}{2}  \right]}{n-p} , \\
\textrm{BIC}(p) &= \ln |{\hat \Dada} | + 
\frac{\left[ pd^2 +\binom{d}{2}  \right] \ln (n-p)}{n-p} =
\sum_{j=1}^d \ln {\hat \delta}_j + 
\frac{\left[ pd^2 +\binom{d}{2}  \right] \ln (n-p)}{n-p} , \\
\textrm{HQ}(p) &= \ln |{\hat \Dada} | + 
\frac{2\left[ pd^2 +\binom{d}{2}  \right] \ln (\ln (n-p))}{n-p} =
\sum_{j=1}^d \ln {\hat \delta}_j  + 
\frac{2\left[ pd^2 +\binom{d}{2}  \right] \ln (\ln (n-p))}{n-p}  ,
\end{aligned}
$$
where $\hat \Dada$ is the estimate of the error covariance matrix $\Dada$.

The AICC (Akaike Information Criterion Corrected) is a bias-corrected version 
of Akaike's AIC, which is an estimate of the Kullback-Leibler index of the 
fitted model relative to the true model and needs further explanation. Here
$$
\textrm{AICC}(p) = -2\ln L ({\hat \A}, {\hat \B}_1 ,\dots ,{\hat \B}_p ,
 {\hat \Dada } )+\textrm{penalty} (p),
$$
where the first term is $-2$ times the log-likelihood function, evaluated at 
the parameter estimates of Theorems~\ref{tetel1} and~\ref{tetel2}, 
whereas the second term   penalizes the  computational complexity. 
The model parameters $\A, \B_1 ,\dots ,\B_p$, and $\Dada$ are estimated by the 
block Cholesky decomposition of the estimated inverse covariance matrix
$\mathfrak{C}_{p+1}^{-1}$  of 
the Gaussian random vector $(\X_{t}^T ,\X_{t-1}^T ,\dots ,\X_{t-p}^T )^T$,
see Algorithms~\ref{alg1} and~\ref{alg2}.
This is a moment estimation, but since our underlying distribution is
multivariate Gaussian, which belongs to the exponential family, asymptotically,
it is also an MLE (for ``large'' $n$) that satisfies the moment matching 
equations, see~\cite{WainwrightJordan}. Of course, the matrices
${\hat \A}, {\hat \B}_1 ,\dots , {\hat \B}_p$, and ${\hat \Dada}$ also 
depend on $p$,
 but for simplicity, we do not denote this dependence. More exactly,
$$
\begin{aligned}
 L ({\hat \A}, {\hat \B}_1 ,\dots ,{\hat \B}_p ,{\hat \Dada} ) &= 
 (2\pi )^{-\frac{(n-p)d}2} |{\hat \Dada } |^{-\frac{n-p}2}
 e^{-\frac12\sum_{t=p+1}^n \U_t^T {\hat \Dada}^{-1} \U_t }  \\ 
&= (2\pi )^{-\frac{(n-p)d}2} (\prod_{j=1}^d {\hat \delta}_j )^{-\frac{n-p}2}
 e^{-\frac12\sum_{t=p+1}^n \sum_{j=1}^d (U_{tj}^2 / {\hat \delta}_j ) }  ,
\end{aligned}
$$
where  
$$ 
\U_t ={\hat \A} (\X_t -{\hat \X}_t ) ,
$$
and
$$
 {\hat \X}_t = -{\hat \A}^{-1} {\hat \B}_1 \X_{t-1} - \dots - 
  {\hat \A}^{-1} {\hat \B}_{p} \X_{t-p } ,
$$
for $t=p+1 , \dots ,n$.

In the unrestricted model, the complexity term (see~\cite{Brockwell}) is
$$
\frac{2 \left[ pd^2 +\binom{d}{2}  \right] (n-p)d}{(n-p)d -pd^2 -
  \binom{d}{2} -1}.
$$
Therefore, 
$$ 
\begin{aligned}
&\textrm{AICC}(p) 
=(n-p)d \ln (2\pi ) +(n-p) \ln |{\hat \Dada} | +
\sum_{t=p+1}^n \U^T {\hat \Dada}^{-1} \U_t +
\frac{2 \left[ pd^2 +\binom{d}{2}  \right] (n-p)d}{(n-p)d -pd^2 -\binom{d}{2} -1} \\
&=(n-p)d \ln (2\pi ) +(n-p) \sum_{j=1}^d \ln {\hat \delta }_j  +
\sum_{t=p+1}^n \sum_{j=1}^d \frac{U_{tj}^2 }{{\hat \delta}_j} +
\frac{2 \left[ pd^2 +\binom{d}{2}  \right] (n-p)d}{(n-p)d -pd^2 -\binom{d}{2} -1} .
\end{aligned}
$$

In the restricted model, the penalization term depends on the cardinalities of the cliques $C_1 ,\dots ,C_k$ and those of the  separators $S_2 ,\dots ,S_k$ that are the same for all $p$. The penalization terms for the four criteria are
$$
\begin{aligned}
\textrm{penalty}_{\textrm{AIC}} (p)&= \frac{2 \left[ pd^2 +\sum_{j=1}^k \binom{|C_j |}{2} +\sum_{j=2}^k \binom{|S_j |}{2}   \right]}{n-p} , \\
\textrm{penalty}_{\textrm{BIC}} (p)&= \frac{\left[ pd^2 +\sum_{j=1}^k \binom{|C_j |}{2} +\sum_{j=2}^k \binom{|S_j |}{2}   \right] \ln (n-p)}{n-p} , \\
\textrm{penalty}_{\textrm{HQ}} (p)&= \frac{2\left[ pd^2 +\sum_{j=1}^k \binom{|C_j |}{2} +\sum_{j=2}^k \binom{|S_j |}{2}    \right] \ln (\ln (n-p))}{n-p}\\
\textrm{penalty}_{\textrm{AICC}} (p)&=
 \frac{2 \left[ pd^2 + \sum_{j=1}^k \binom{|C_j |}{2} +
  \sum_{j=2}^k \binom{|S_j |}{2}  \right] (n-p)d}{(n-p)d -pd^2 -
   \sum_{j=1}^k \binom{|C_j |}{2} - \sum_{j=2}^k \binom{|S_j |}{2} -1} .
\end{aligned}
$$
The cliques are usually of ``small'' sizes that can reduce 
computational complexity, but only when the number of variables $d$ is 
much ``larger'' than the clique sizes. In the financial model, with merely 
$d=8$ variables, the number of parameters is not smaller in the restricted case
than in the unrestricted one, but the gain can be substantial with $d$
larger than, say, 100. Also this number of parameters to be estimated in
the decomposable model is the same as the number of product moments calculated
for the covariance selection, only  within the cliques and separators, 
see Equation~\eqref{Kestimated}.

All of these criteria are tested, for both the restricted and unrestricted 
CVAR$(p)$ models, using the financial data above for $p=1,2,\dots ,9$.
The results for the unrestricted case are shown in Table~\ref{ofu}.

\begin{table}[ht]
	\centering
	\scriptsize
	\resizebox{0.6\columnwidth}{65pt}{%
	\begin{tabular}{r cccccccc}
		$p$ & \textbf{AIC} & \textbf{AICC} & \textbf{BIC} & \textbf{HQ} \\ 
		1 & -76.81 & \textbf{-33222.68} & \textbf{-76.07} & \textbf{-76.52} \\
		2 &	\textbf{-76.85} & -33173.98 & -75.60 & -76.36\\
		3  & -76.84 & -33095.75 & -75.08 & -76.15 \\ 
		4  & -76.83 & -33011.98 & -74.55 & -75.94 \\ 
		5 & -76.77 & -32893.23 & -73.97 & -75.67 \\ 
		6 & -76.69 & -32766.33 & -73.37 & -75.39 \\ 
		7 & -76.58 & -32612.38 & -72.74 & -75.08 \\ 
		8 & -76.48 & -32457.38 & -72.11 & -74.77 \\
		9 & -76.41 & -32316.33 & -71.52 & -74.49 \\ 
	\end{tabular}%
	}
	\caption{Order selection criteria for the unrestricted Financial
CVAR$(p)$ model, bold-face values represent minimum of each criterion.}
 \label{ofu}
\end{table}

Observe that in the unrestricted case, 
AIC reaches the minimum for $p=2$, whereas AICC,  BIC, and HQ for 
$p=1$. This is in accord with our previous experience that the parameter 
matrices did not change much after the first or second day.


\begin{table}[ht]
	\centering
	\scriptsize
	\resizebox{0.6\columnwidth}{65pt}{%
	\begin{tabular}{r cccccccc}
		$p$ & \textbf{AIC} & \textbf{AICC} & \textbf{BIC} & \textbf{HQ} \\ 
		1 & -76.82 & \textbf{-33224.46} & \textbf{-76.02} & \textbf{-76.51} \\
		2 & -76.85 & -33175.16 & -75.55 & -76.34\\
		3  & -76.88 & -33114.07 & -75.06 & -76.17 \\ 
		4  & \textbf{-76.94} & -33068.80 & -74.61 & -76.03 \\ 
		5 & -76.89 & -32952.72 & -74.03 & -75.77 \\ 
		6 & -76.86 & -32852.10 & -73.49 & -75.54 \\ 
		7 & -76.76 & -32700.55 & -72.86 & -75.23 \\ 
		8 & -76.75 & -32594.18 & -72.32 & -75.01 \\
		9 & -76.72 & -32476.45 & -71.78 & -74.79 \\ 
	\end{tabular}%
	}
	\caption{Order selection criteria for the restricted Financial
CVAR$(p)$ model, bold-face values represent minimum of each criterion.}
  \label{ofr}
\end{table}

In the restricted case (see Table~\ref{ofr}), 
except for the AIC, every criteria suggests that the 
best model is obtained with $p = 1$.
So the parameter matrices did not change much after the first day, except for
 AIC, which seems to overspecify the model and was the lowest on the 4th day, i.e. the last workday after the first workday. 


\subsection{IMR (Infant Mortality Rate) Longitudinal Data}

Here, we used the longitudinal data of six indicators (components of $\X_t$), 
spanning 21 years ($1995-2015$) from the World Bank in case of Egypt:
\begin{equation}\label{imr}
\begin{aligned}
&1: \, \textrm{IMR (Infant mortality rate)},  \\
&2: \, \textrm{MMR (Maternal mortality ratio)}, \\
&3: \, \textrm{HepB (Hepatitis-B immunization)}, \\
&4: \, \textrm{GDP (Gross domestic per capita)}, \\
&5: \, \textrm{OPExp (Out-of-pocket health expenditure as \% of HExp)}, \\
&6: \, \textrm{HExp (Current health expenditure as \% of GDP)}.
\end{aligned}
\end{equation}
For more details about these indicators, see~\cite{Abdelkhalek}. 
Through the VAR($p$) model, we show the contemporaneous and lag time effects between the components. Since the sample size is small, 
we investigate only VAR($1$) model in the unrestricted and restricted 
situations. Further, the variables are measured on different scales; so, 
we use the autocorrelations which are the autocovariances of the standardized variables. We distinguish between two working hypotheses with respect to two different ordering of the variables given by an expert:
\begin{itemize}
\item 
Case 1: $\{ \textrm{IMR,MMR,HepB,OPExp,HExp ,GDP} \}$.
\item
Case 2: $\{ \textrm{IMR,MMR,HepB,GDP,OPExp , HExp} \}$.
\end{itemize}

In the unrestricted CVAR(1) model, both orderings work, 
but we present only Case 1. 
(The estimated matrices $\A$ and $\B$ are mostly the same in both cases 
but the entries are interchanged with respect to the ordering of the variables.)
The entries of matrix $\A$ (see Table~\ref{A_un1}) represent the 
contemporaneous effects (path coefficients) between the components at time $t$. The MMR has the largest contemporaneous inverse causal effect on the IMR, 
i.e. an increase in the MMR 
caused a decrease in the IMR by $1.13$. Matrix $\B$ (see Table~\ref{B_un1}), 
on the other hand, indicates the path coefficients of the one time lag causal effect of $\X_{t-1}$ on the current $\X_{t}$ components. An increase in the IMR at one year time lag caused an increase in the IMR at the current time by $0.29$.
All other path coefficients 
in the matrices $\A, \B$ can be explained likewise.

\begin{table}[ht]
	\begin{tabular}{r c c c c c c}
		{} &  IMR &     MMR &    HepB &   OPExp &    HExp &     GDP \\
		IMR   &  1.0 & -1.1259 & -0.0161 &  0.0003 &  0.0176 & -0.1348 \\
		MMR   &  0.0 &  1.0000 &  0.3594 &  0.0492 & -0.0684 &  0.7135 \\
		HepB  &  0.0 &  0.0000 &  1.0000 & -0.1626 &  0.2510 & -0.8196 \\
		OPExp &  0.0 &  0.0000 &  0.0000 &  1.0000 & -0.6876 & -0.4229 \\
		HExp  &  0.0 &  0.0000 &  0.0000 &  0.0000 &  1.0000 &  0.6749 \\
		GDP   &  0.0 &  0.0000 &  0.0000 &  0.0000 &  0.0000 &  1.0000 \\
	\end{tabular}
	\caption{$\A$ matrix for the IMR unrestricted VAR($1$) model of Case 1 (rounded to $4$ decimals).}
	\label{A_un1}
\end{table}

\begin{table}[ht]
	\begin{tabular}{r c c c c c c}
		{} &    IMR-1 &    MMR-1 &  HepB-1 &  OPExp-1 &  HExp-1 &   GDP-1 \\
		IMR   &   0.2986 &  -0.3589 & -0.0076 &   0.0042 & -0.0115 & -0.0639 \\
		MMR   &  -0.0149 &  -0.7469 & -0.2358 &   0.0540 &  0.0193 & -0.5577 \\
		HepB  &  13.0541 & -15.7658 & -1.1915 &  -0.3506 &  0.2170 & -1.7902 \\
		OPExp &   7.1616 &  -8.0906 & -0.2994 &  -0.1038 & -0.0215 & -0.7720 \\
		HExp  &   1.6861 &  -2.9922 & -0.4650 &  -0.1566 & -0.0681 & -1.0913 \\
		GDP   & -11.0674 &  13.2182 &  0.3254 &   0.3204 & -0.3099 &  1.2129 \\
	\end{tabular}
	\caption{$\B$ matrix for the IMR unrestricted VAR($1$) model of Case 1 (rounded to $4$ decimals).}
\label{B_un1}
\end{table}

In the restricted CVAR(1) model, the graph structure is important. 
We consider only Case $2$ that provides the RZP and corresponds
to the ordering of~\eqref{imr}.
Note that the so obtained DAG is Markov equivalent to its undirected skeleton. 
The decomposable structure of the JT has two cliques
and  only one separator as follows:
$$
\begin{aligned}
C_1 &= \{ \textrm{IMR,MMR,HepB3,GDP,HExp} \} =\{1,2,3,4,6 \}, \\
C_2 &= \{ \textrm{OPExp,HExp} \}=\{ 5,6 \} ,\\
S_2 &=  \{ \textrm{HExp} \} =\{ 6 \}.
\end{aligned}
$$
In this case, the $n-1$ element sample including lag 1 variables 
is used to estimate the $12\times 12$ matrix $\hat\K$ with
covariance selection, see Equation~\eqref{khat}. 
Then the LDL algorithm was applied to the
obtained $\hat\K$ 
to estimate the model parameters $\A, \B$. Unlike the unrestricted model, here there are prescribed zero entries in $\hat\K$ and $\A$. Specifically, the zeros of the left upper $6 \times 6$ corner of $\hat\K$
  will necessarily
result in the zeros of the estimated matrix $\A$ in the same positions.
Similarly to the unrestricted situation, the non-zero 
upper-diagonal entries  of $\A$ (see Table~\ref{A_rs}) represent the path coefficients of the contemporaneous causal effects of $\X_t$, while the entries of
the matrix $\B$  (see Table~\ref{B_rs})
represent the 1 time lag causal effects of the $\X_{t-1}$ components 
on the $\X_t$ components.

\begin{table}[ht]
	\begin{tabular}{r c c c c c c}
		{} &  IMR &     MMR &    HepB &     GDP &  OPExp &    HExp \\
		IMR   &  1.0 &  0.0736 &  0.0052 &  0.0111 &    0.0 &  0.0040 \\
		MMR   &  0.0 &  1.0000 &  0.0237 &  0.1180 &    0.0 & -0.0116 \\
		HepB  &  0.0 &  0.0000 &  1.0000 & -0.3418 &    0.0 &  0.0236 \\
		GDP   &  0.0 &  0.0000 &  0.0000 &  1.0000 &    0.0 & -0.0035 \\
		OPExp &  0.0 &  0.0000 &  0.0000 &  0.0000 &    1.0 & -0.7854 \\
		HExp  &  0.0 &  0.0000 &  0.0000 &  0.0000 &    0.0 &  1.0000 \\
	\end{tabular}
	\caption{$\A$ matrix for the IMR restricted VAR($1$) model of Case 2 (rounded to $4$ decimals).}
	\label{A_rs}
\end{table}

\begin{table}[ht]
	\begin{tabular}{r c c c c cc}
		{} &   IMR-1 &   MMR-1 &  HepB-1 &   GDP-1 &  OPExp-1 &  HExp-1 \\
		IMR   & -0.9198 & -0.0592 & -0.0021 &  0.0053 &  -0.0013 & -0.0049 \\
		MMR   & -1.0175 &  0.2511 & -0.0002 &  0.0607 &  -0.0047 &  0.0039 \\
		HepB  &  3.6850 & -4.4606 & -0.9058 & -0.4230 &  -0.0811 & -0.0950 \\
		GDP   &  0.5849 & -0.4453 &  0.0640 & -0.8573 &   0.0896 &  0.1086 \\
		OPExp &  3.6432 & -3.7486 & -0.1413 & -0.4380 &   0.0273 & -0.0926 \\
		HExp  &  1.9561 & -3.4687 & -0.5221 & -0.6302 &  -0.2298 & -0.1171 \\
	\end{tabular}
	\caption{$\B$ matrix for the IMR restricted VAR($1$) model of Case 2 (rounded to $4$ decimals).}
\label{B_rs}
\end{table}

\section{Discussion}\label{discussion}

The main contribution of our paper is the introduction of 
causality in
VAR models by using graphical modeling tools. 
SVAR models are known in the literature, 
but there the inclusion of the upper triangular matrix $\A$ 
rather facilitates an alternative
solution for the Yule--Walker equations, and not the causal ordering of the
contemporaneous effects. 

Our unrestricted CVAR model does this job, where the recursive ordering of the 
variables follows a DAG ordering in the directed graphical model 
contemporaneously and the entries of $\A$ are treated like
path coefficients of SEM. 
Also, the white noise process $\U_t$ of structural shocks 
(see Equation~\eqref{mimop}) is obtained from
the process $\V_t$ of innovations in the reduced form 
(see Equation~\eqref{reduced}) 
and have an econometric interpretation.
The structural shocks are mutually uncorrelated and they are assigned
to the individual variables. 
They also represent unanticipated changes in the observed econometric variables.
However, they
are not just orthogonalized innovations, but here the labeling of the nodes
and the graph skeleton behind the matrix  $\A$ also counts.

In the unrestricted case, the following estimation scheme is used.
The DAG is built partly by expert knowledge and partly
by starting with an undirected Gaussian graphical model, using known
algorithms (e.g. MCS) to find a triangulated graph and a (not necessarily
unique) perfect labeling of the nodes, in which ordering the directed and
undirected models are Markov equivalent to each other (there are no sink V
configurations in the DAG). However, here the Markov equivalence is not
important: even if the undirected graph is not triangulated, and the DAG
contains sink Vs, the DAG ordering (given, e.g. by an expert) can be used
to estimate the $\A$ and $\B$ matrices, which are full in the sense that no
zero constraints for their entries are assumed at the  beginning. 
After having the DAG ordering, we apply the block LDL decomposition for
the estimated block matrix $\mathfrak{C}_{2}$ or $\mathfrak{C}_{p+1}$,
and retrieve the estimated parameter matrices by Theorem~\ref{tetel1} 
or~\ref{tetel2}.

It is the restricted CVAR model, where zero constraints for the entries of $\A$
(in the given DAG ordering) are made. For this purpose, we re-estimate the
covariance matrix (the big block matrix, the size of which depends on the 
order $p$ of the model) such that the entries in the left upper block of its 
inverse are zeros in the no-edge positions. For this, there is the method
of covariance selection at our disposal, which
works for Gaussian variables even if the prescribed zeros in the inverse 
covariance matrix do not have the RZP 
(RZP is just the property of decomposable models). 
In this case, our algorithm first applies algorithms (e.g. MCS) to find the
JT structure of the graph (which is equivalent to having
an RZP). The estimation scheme is enhanced with
covariance selection, for which there are closed form estimates in the
decomposable case, see Appendix~\ref{undirg}.
Actually, we use  an improved version of the
covariance selection that needs higher order autocovariances too.
Since the necessary product moment estimates include only variables belonging
to the cliques and separators, it can reduce the computational complexity of
the restricted CVAR model compared to the unrestricted one. However, it holds
only if the number of variables $d$ is much larger than the clique sizes.
The information criteria, applied to select the optimal order $p$, 
also  take into consideration the number of parameters to be estimated.

\section{Conclusions and Further Perspectives}\label{concl}

Our algorithm is as well applicable to longitudinal data instead of time
series.
 The $p=0$ case resolves the problem posed in~\cite{Wermuth}, and the $p=1$
case is also applicable to solve a SEM with 
endogenous and exogenous variables.

As a further perspective, lagged causalities could also be introduced, with
some upper triangular matrices $\B$. For example, if the previous time
observations influence the present time ones, and the order of causalities
is the same as that of the contemporaneous ones, then $\B_1$ is also upper
triangular. This problem can be solved by running the block Cholesky
decomposition with $2d$ singleton blocks and treating
only the other blocks ``en block''.


\vfill

\noindent\textbf{Author Contributions:}
{The theoretical parts of the paper with theorems and proofs were written by 
Marianna Bolla. The Python code
for the algorithms was written by Haoyou Wang and Renyuan Ma, and tested with 
the simulated data of William Thompson and Catherine Donner. Dongze Ye 
developed the graph-building and covariance selection program in Python, and organized the Python 
programs in a notebook with many explanations. Valentin Frappier wrote
the part of the Python program for order selection.
 M\'at\'e Baranyi PhD, former PhD student of the first author,
gave useful ideas during the initial phase of the research, and provided 
continuous technical support later on.
Fatma Abdelkhalek PhD, also former PhD student of the first author, 
 verified all the proof computations numerically, continuously worked on the 
manuscript and provided the World Bank data, together with building the graphs.}


\noindent\textbf{Data Availability:}
{The third-party financial dataset analyzed in this article is available in
the UCI Machine Learning Repository, and was collected by the authors 
of~\cite{Akbilgic}.
The dataset is available in: Dua, D. and Graff, C. (219).  
UCI Machine Learning Repository, Irvine, CA: University of California, School of
Information and Computer Science, 
\url{https://archive.ics.uci.edu/ml/datasets/ISTANBUL+STOCK+EXCHANGE}.

The World Bank Data on infant mortality rate are available on
\url{https://data.worldbank.org/indicator}, also see~\cite{Abdelkhalek}.} 

\noindent\textbf{Acknowledgments:}
{The research was done under the auspices of the Budapest Semesters in Mathematics program, in the framework of an undergraduate online research course in
summer 2021, with the participation of US undergraduate students.
Two PhD students of the corresponding author  also participated. 
In particular, Fatma Abdelkhalek's work was funded by a scholarship under the 
Stipendium Hungaricum program between Egypt and Hungary;  whereas,
Valentin Frappier's internship by the Erasmus program of the EU.}



\vbox{
\noindent\textbf{Abbreviations}\\
{
The following abbreviations are used in this manuscript:\\
\noindent 
\begin{tabular}{@{}ll}
VAR & Vector Auto-Regression  \\
SVAR & Structural Vector Auto-Regression   \\
CVAR & Causal Vector Auto-Regression \\
SEM & Structural Equation Modeling \\
DAG & Directed Acyclic Graph \\
JT & Junction Tree \\
MCS & Maximal Cardinality Search \\
IPS & Iterative Proprtional Scaling \\
RZP & Reducible Zero Pattern \\
MRF & Markov Random Field  \\
AIC & Akaike Information Criterion \\
AICC & Akaike Information Criterion Corrected \\
BIC & Bayesian Information Criterion \\
HQ & Hannan and Quinn's criterion \\
MLE & Maximum Likelihood Estimate \\
PLS & Partial Least Squares regression \\
RMSE & Root Mean Square Error \\
IMR & Infant Mortality Rate  \\
LDL & variant of the Cholesky decomposition for a symmetric, positive \\
\textrm{   } & 
semidefinite matrix as $\LL$(lower triangular)$\times \DD$(diagonal)$\times \LL^T$ 
\end{tabular}
}}

\begin{appendices}
\numberwithin{equation}{section}

\section{Proofs of the main theorems}\label{appA}

\subsection{Proof of Theorem~\ref{tetel1}}\label{biz1}

First of all note that the block Cholesky decomposition
applies to $\K$ partitioned symmetrically into
$(d+1)\times (d+1)$ blocks of sizes $1,\dots ,1,d$, where the number of
singleton blocks (of size 1) is  $d$. (In the $p=0$ case all the blocks are
singletons, so the standard LDL decomposition is applicable.)
Therefore, in 
the main diagonal of the resulting $\LL$ we have number $d$ of 
1s and $\I_{d\times d}$ (the $d\times d$ identity matrix), 
see the forthcoming Equation~\eqref{uj}.
In other words, the last $d$ rows (and columns) are treated ``en block'',
this is why here indeed the block LDL (variant of the block Cholesky) 
decomposition is applicable.

Let us compute the inverse of the matrix $\LL \DD \LL^T$ with block matrices
$\LL$ and $\DD$
partitioned as in Equation~\eqref{fatma}. 
For the time being we only assume that $\A$ is a
$d\times d$ upper triangular matrix  with 1s along its main diagonal, 
$\B$ is $d\times d$, and the diagonal matrix $\Dada$ has positive
diagonal entries. We will use
the computation rule of the inverse of a
symmetrically partitioned block matrix~\cite{Rozsa}, which is applicable 
due to the fact that $|\A | =1$, so the matrix $\A$ is invertible:
$$
\begin{aligned}
 (\LL \DD \LL^T )^{-1} &=\begin{pmatrix}
  \A & \B \\
  \OO_{d\times d}  &\I_{d\times d}
 \end{pmatrix}^{-1}
   \begin{pmatrix}
  \Dada^{-1}  &\OO_{d\times d} \\
  \OO_{d\times d}  &\C^{-1} (0)
 \end{pmatrix}^{-1}
\begin{pmatrix}
  \A^T  &\OO_{d\times d} \\
  \B^T  &\I_{d\times d}
 \end{pmatrix}^{-1}   \\
&=\begin{pmatrix}
  {\A}^{-1}   &  -{\A}^{-1} \B \\
  \OO_{d\times d}  &\I_{d\times d}
 \end{pmatrix}
 \begin{pmatrix}
  \Dada  &\OO_{d\times d} \\
  \OO_{d\times d}  &\C (0)
 \end{pmatrix}
\begin{pmatrix}
  (\A^T)^{-1}  & \OO_{d\times d} \\
 -\B^T (\A^T)^{-1}  &\I_{d\times d}
\end{pmatrix}  \\
&= \begin{pmatrix}
  \A^{-1} \Dada ({\A}^{-1})^T +\A^{-1} \B \C (0) \B^T ({\A}^{-1})^T  
 & -{\A}^{-1} \B\C (0) \\
  -\C (0)  \B^T (\A^T)^{-1}  &\C (0) 
 \end{pmatrix} .
\end{aligned}
$$
Now we are going to prove that the above matrix equals 
$\mathfrak{C}_2$ if and only if $\A ,\B ,\Dada$ satisfy the model equations.
Comparing the blocks to those of~\eqref{egyik}, 
the right bottom block is $\C (0)$ in both expressions. 
Comparing the left bottom blocks, we get
$-\C(0) \B^T (\A^T)^{-1} =\C (1)$, and so, $\B^T = -\C^{-1} (0) \C (1) \A^T$ 
and $\B = -\A \C^T (1) \C^{-1} (0)$ should hold for  $\B$. It is 
in accord with the model equation. Indeed, \eqref{mimo} 
is  equivalent to
$$
 \B \X_{t-1} = -\A \X_{t} +\U_t  ,
$$
which, after multiplying with $\X_{t-1}^T$ from the right and 
taking expectation, yields
$\B \C (0) = -\A \C^T (1)$
that in turn provides 
\begin{equation}\label{B}
 \B = -\A \C^T (1) \C^{-1} (0) .
\end{equation}
By symmetry, the same applies to the right upper block.
As for the left upper block,
$$
 \A^{-1} \Dada ({\A}^{T})^{-1} +\A^{-1} \B \C (0) \B^T ({\A}^{T})^{-1} =\C(0)
$$
should hold. Multiplying this equation with $\A$ from the left and with $\A^T$
from the right, we get the equivalent equation
\begin{equation}\label{dada1}
 \Dada = \A \C(0)\A^T -\B \C (0) \B^T .
\end{equation}
This is in accord with Equation~\eqref{mimo} that implies
$$
 \E (\A \X_{t}  +\B \X_{t-1} )(\A \X_{t} +\B \X_{t-1} )^T =
 \A \C (0) \A^T +\A \C^T (1) \B^T + \B \C(1) \A^T + \B \C (0) \B^T =\Dada  .
$$
Combining this with Equation~\eqref{B}, we get
$$
\begin{aligned}
 \Dada &=
 \A \C (0) \A^T +\A \C^T (1) \B^T + \B \C(1) \A^T + \B \C (0) \B^T  \\
 &= \A \C (0) \A^T -\A \C^T (1) \C^{-1} (0) \C (1) \A^T  
  - \A \C^T (1) \C^{-1} (0) \C(1) \A^T   \\
 &+ \A \C^T (1) \C^{-1} (0) \C (0) \C^{-1} (0) \C(1) \A^T \\
 &=\A \C (0) \A^T -\A \C^T (1) \C^{-1} (0) \C (1) \A^T 
  = \A \C(0)\A^T -\B \C (0) \B^T  ,
\end{aligned}
$$
that also satisfies~\eqref{dada1}.

Summarizing, we have proved that under the model equations,
$(\LL \DD \LL^T )^{-1} = \mathfrak{C}_2$, or equivalently,
$\LL \DD \LL^T  = \K$ indeed holds. 
In view of the uniqueness of the block LDL decomposition
(under positive definiteness of the involved matrices), 
this finishes the proof.

\subsection{Algorithm for the Block LDL Decomposition of Section~\ref{biz1}}\label{alg1}

By the preliminary assumptions, $\K$ and so $\DD$ are positive definite;
therefore, $\Dada$ has positive diagonal entries. 
To apply the protocol of the block Cholesky decomposition
which gives the theoretically guaranteed unique solution,
it is worth writing the above matrices according to the blocks as follows.
The matrix $\LL$ has the partitioned form
\begin{equation}\label{uj}
\LL = \begin{pmatrix}
  1      &0       &0      &\dots  &0 &0 &\dots &0 \\
 {\ell}_{21}  &1       &0       &\dots  &0 &0 &\dots &0 \\
 {\ell}_{31}  &{\ell}_{32}   &1      &\dots  &0 &0 &\dots &0 \\
 \vdots  &\vdots  &\vdots &\vdots  &\vdots &0 &\dots &0 \\
 {\ell}_{d1}  &{\ell}_{d2} &\dots &{\ell}_{d,d-1}  &1 &0 &\dots &0 \\
 \bm{\ell}_{d+1,1}  &\bm{\ell}_{d+1,2} &\vdots &\bm{\ell}_{d+1,d-1} 
 &\bm{\ell}_{d+1,d} & &\I_{d\times d} & 
\end{pmatrix} ,
\end{equation}
where the $2d\times 2d$ lower triangular matrix $\LL$ is also lower triangular
with respect to its blocks which are partly scalars, partly vectors, 
partly matrices as follows: 
$$
 {\ell}_{ij} = \left\{\begin{array}{ll}
  a_{ji},  & j=1,\dots ,d-1; \quad i=j+1, \dots ,d;  \\
  1 & i=j=1,\dots ,d  ; \\
  0 & i=1,\dots ,d; \quad j=i+1,\dots ,2d;
                 \end{array}
                 \right.
$$
further, the vectors $\bm{\ell}_{d+1,j}$ are $d\times 1$ for $j=1,\dots ,d$, 
and 
comprise the column vectors of the $d\times d$ matrix $\B^T$. The matrix
in the bottom right block is the $d\times d$ identity $\I_{d\times d}$, and 
above it, the
zero entries can be arranged into the $d\times d$ zero matrix $\OO_{d\times d}$.

The $2d\times 2d$ block-diagonal matrix $\DD$ in partitioned form is
$$
\DD = \begin{pmatrix}
  \delta_1^{-1}      &0       &0      &\dots  &0 &0 &\dots &0 \\
  0             &\delta_2^{-1}       &0       &\dots  &0 &0 &\dots &0 \\
  0             &0   &\delta_3^{-1}      &\dots  &0 &0 &\dots &0 \\
 \vdots  &\vdots  &\vdots &\vdots  &\vdots &\0 &\dots &\0 \\
 0  &0 &0 &0  &\delta_d^{-1} &0 &\dots &0 \\
 \0  &\0 &\vdots &\0 
 &\0 & &\C^{-1} (0)  & 
\end{pmatrix} ,
$$
where the $d\times 1$ vectors $\0$ comprise $\OO_{d\times d}$ in the left 
bottom, and
the entries comprise the inverse of  the $d\times d$ positive 
definite matrix $\C (0)$ in the right
bottom block. We perform the following multiplications
of block matrices, also using formulas~\cite{Golub,Rozsa} for their inverses
and the algorithm proposed in~\cite{Nocedal} to get 
the recursion of the block LDL decomposition that goes on as follows.
\begin{itemize}
\item Outer cycle (column-wise). 
 For $j=1,\dots ,d$: 
 $\delta^{-1}_j = k_{jj} -\sum_{h=1}^{j-1} {\ell}_{jh} \delta_h^{-1} 
 {\ell}_{jh}$ (with the reservation that $\delta^{-1}_1 = k_{11}$);
\item Inner cycle (row-wise). For $i=j+1 ,\dots ,d$:
\begin{equation}\label{*}
 {\ell}_{ij} = \left( k_{ij} - \sum_{h=1}^{j-1} {\ell}_{ih} \delta_h^{-1}  
 {\ell}_{jh} \right) \delta_j 
\end{equation}
and   
$$
 \bm{\ell}_{d+1,j} = \left( \bm{k}_{d+1,j} - \sum_{h=1}^{j-1} 
 \bm{\ell}_{d+1,h} \delta_h^{-1}  {\ell}_{jh} \right) \delta_j 
$$
(with the reservation that in the $j=1$ case the summand is zero),
where  $\bm{k}_{d+1,j}$ for $j=1,\dots ,d$ is $d\times 1$ vector 
in the bottom left block of $\K$. 
\end{itemize}
Note that the last step of the outer cycle, when $j=d+1$, formally would be
$$
 \C^{-1} (0) = \K_{d+1,d+1} -\sum_{h=1}^{d} \bm{\ell}_{d+1,h} \delta_h^{-1} 
  \bm{\ell}_{d+1, h}^T  =
 \K_{d+1,d+1} -\sum_{h=1}^{d}  \delta_h^{-1} \bm{\ell}_{d+1,h}
  \bm{\ell}_{d+1, h}^T ,
$$
where  $\bm{\ell}_{d+1,h}$ for $h=1,\dots ,d$ are $d\times 1$ vectors and
$\K_{d+1,d+1}$ is the bottom right $d\times d$ block of the $2d\times 2d$ 
concentration matrix $\K$; but it need not be performed as it is in accord 
with Theorem~\ref{tetel1}.
 Then no inner cycle follows and the recursion
ends in one run.

It is obvious that  the above decomposition has a nested structure, so
for the first $d$ rows of $\LL$, 
only its previous rows or preceding entries in the same row enter into
the calculation, as if we performed the standard LDL decomposition of $\K$.
 Therefore,
${\ell}_{ij} = a_{ji}$ for $j=1,\dots ,d-1$, $i=j+1, \dots ,d$ that are
the partial regression coefficients akin to those offered by 
the standard LDL decomposition 
$\K = {\tilde \LL } {\tilde \DD } {\tilde \LL }^T$; so the first $d$ rows of
$\tilde \LL$ and $\LL$ are the same, and the first $d$ rows of
$\tilde \DD$ and $\DD$ are the same too.

When the process terminates after finding the first $d$ rows of $\LL$, 
we consider the blocks ``en block'' and get the matrix
$\B =( \bm{\ell}_{d+1,1}, \dots ,  \bm{\ell}_{d+1,d})^T$.

\subsection{Proof of Theorem~\ref{tetel2}}\label{biz2}

Note that here the block Cholesky decomposition
applies to $\K$ partitioned symmetrically into
$(d+1)\times (d+1)$ blocks of sizes $1,\dots ,1,pd$ with number $d$ of 
singleton blocks. 
(Therefore, in 
the main diagonal of $\LL$ we have number $d$ of 
1s and $\I_{pd\times pd}$.)
The $d\times pd$ matrix $\B$, transpose of $\B^T$, will
contain the coefficient 
matrices of Equation~\eqref{mimop} in its blocks, like
$$
  \B = (\B_1 \dots \B_p ) .
$$ 

The proof goes on similarly as in Section~\ref{biz1}. 
Though, for completeness and being able to formulate the algorithm, we discuss
it herein.
Let us compute the inverse of the matrix $\LL \DD \LL^T$ with block matrices
$\LL$ and $\DD$
partitioned as in Equation~\eqref{fatma1}. 
For the time being we only assume that $\A$ is a
$d\times d$ upper triangular matrix  with 1s along its main diagonal, 
$\B$ is $d\times pd$, and the diagonal matrix $\Dada$ has positive
diagonal entries. We can again use
the computation rule of the inverse of 
symmetrically partitioned block matrices, since the matrix $\A$ is invertible.  
$$
\begin{aligned}
 (\LL \DD \LL^T )^{-1} &=\begin{pmatrix}
  \A & \B \\
  \OO_{pd\times d}  &\I_{pd\times pd}
 \end{pmatrix}^{-1}
   \begin{pmatrix}
  \Dada^{-1}  &\OO_{d\times pd} \\
  \OO_{pd\times pd}  &\mathfrak{C}_p^{-1} 
 \end{pmatrix}^{-1}
\begin{pmatrix}
  \A^T  &\OO_{d\times pd} \\
  \B^T  &\I_{pd\times pd}
 \end{pmatrix}^{-1}   \\
&=\begin{pmatrix}
  {\A}^{-1}   &  -{\A}^{-1} \B \\
  \OO_{pd\times d}  &\I_{pd\times pd}
 \end{pmatrix}
 \begin{pmatrix}
  \Dada  &\OO_{d\times pd} \\
  \OO_{pd\times d}  &\mathfrak{C}_p 
 \end{pmatrix}
\begin{pmatrix}
  (\A^T)^{-1}  & \OO_{d\times pd} \\
 -\B^T (\A^T)^{-1}  &\I_{pd\times pd}
\end{pmatrix}  \\
&= \begin{pmatrix}
  \A^{-1} \Dada ({\A}^{-1})^T +\A^{-1} \B \mathfrak{C}_p \B^T ({\A}^{-1})^T  
 & -{\A}^{-1} \B\mathfrak{C}_p \\
  -\mathfrak{C}_p   \B^T (\A^T)^{-1}  &\mathfrak{C}_p  
 \end{pmatrix} .
\end{aligned}
$$
Now we are going to prove that the above matrix equals 
$\mathfrak{C}_{p+1}$ if and only if $\A ,\B ,\Dada$ satisfy the model equations.
Comparing the blocks to those of~\eqref{eq:block_Toeplitz}, 
the right bottom block is $\mathfrak{C}_p$ in both expressions. 
Comparing the left bottom blocks, we get
$-\mathfrak{C}_p \B^T (\A^T)^{-1} =\C (1,\dots ,p)$, and so, $\B^T = 
-\mathfrak{C}_p^{-1} \C (1,\dots ,p) \A^T$ 
and $\B = -\A \C^T (1,\dots ,p) \mathfrak{C}_p^{-1}$ should hold for  $\B$. 
It is in accord with the model equation: indeed, \eqref{mimop} 
is  equivalent to
$$
 \B_1 \X_{t-1} + \dots + \B_{p} \X_{t-p } = -\A \X_{t} +\U_t ,
$$
which, after multiplying with $\X_{t-1}^T, \dots ,\X_{t-p}^T$ from the right and 
taking expectation, in concise form yields
$\B \mathfrak{C}_p = -\A \C^T (1,\dots ,p)$
that in turn provides 
\begin{equation}\label{BB}
 \B = -\A \C^T (1,\dots ,p) \mathfrak{C}_p^{-1} .
\end{equation}
By symmetry, it also applies to the right upper block.
As for the left upper block,
$$
 \A^{-1} \Dada ({\A}^{T})^{-1} +\A^{-1} \B \mathfrak{C}_p \B^T ({\A}^{T})^{-1} =\C(0)
$$
should hold. Multiplying this equation with $\A$ from the left and with $\A^T$
from the right, we get the equivalent equation
\begin{equation}\label{dada}
 \Dada = \A \C (0)  \A^T -\B \mathfrak{C}_p  \B^T .
\end{equation}
This is in accord with Equation~\eqref{mimop} that implies
$$
\begin{aligned}
 &\E (\A \X_t +\B_1 \X_{t-1} + \dots + \B_{p} \X_{t-p })
   (\A \X_{t} +\B_1 \X_{t-1} + \dots + \B_{p} \X_{t-p } )^T  \\
 &=\A \C (0)  \A^T +\A \C^T (1,\dots ,p) \B^T + \B \C(1,\dots ,p) \A^T + 
  \B \mathfrak{C}_p  \B^T =\Dada  .
\end{aligned}
$$
Combining this with Equation~\eqref{BB}, we have
$$
\begin{aligned}
 \Dada &=
 \A \C (0) \A^T +\A \C^T (1,\dots ,p) \B^T + \B \C(1,\dots ,p) \A^T + 
  \B \mathfrak{C}_p \B^T  \\
 &= \A \C (0)   \A^T -\A \C^T (1,\dots ,p) \mathfrak{C}_p^{-1}  
  \C (1,\dots ,p) \A^T  
  - \A \C^T (1,\dots ,p) \mathfrak{C}_p^{-1} \C(1,\dots ,p) \A^T   \\
 &+ \A \C^T (1,\dots ,p) \mathfrak{C}_p^{-1}  \mathfrak{C}_p  
 \mathfrak{C}_p^{-1}  \C(1,\dots ,p) \A^T \\
 &=\A \C (0)  \A^T -\A \C^T (1,\dots ,p) \mathfrak{C}_p^{-1} 
  \C (1,\dots ,p) \A^T 
  = \A \C (0) \A^T -\B \mathfrak{C}_p \B^T ,
\end{aligned}
$$
that also satisfies~\eqref{dada}.

Summarizing, we have proved that under the model equations,
$(\LL \DD \LL^T )^{-1} = \mathfrak{C}_{p+1}$, or equivalently,
$\LL \DD \LL^T  = \K$ indeed holds. 
In view of the uniqueness of the block LDL decomposition
(under positive definiteness of the involved matrices), 
this finishes the proof.

\subsection{Algorithm for the Block LDL Decomposition of Section~\ref{biz2}}\label{alg2}

Again, the protocol of the block Cholesky decomposition is applied to the
involved matrices in block partitioned form. Here
$$
\LL = \begin{pmatrix}
  1      &0       &0      &\dots  &0 &0 &\dots &0 \\
 {\ell}_{21}  &1       &0       &\dots  &0 &0 &\dots &0 \\
 {\ell}_{31}  &{\ell}_{32}   &1      &\dots  &0 &0 &\dots &0 \\
 \vdots  &\vdots  &\vdots &\vdots  &\vdots &0 &\dots &0 \\
 {\ell}_{d1}  &{\ell}_{d2} &\dots &{\ell}_{d,d-1}  &1 &0 &\dots &0 \\
 \bm{\ell}_{d+1,1}  &\bm{\ell}_{d+1,2} &\vdots &\bm{\ell}_{d+1,d-1} 
 &\bm{\ell}_{d+1,d} & &\I_{pd\times pd} & 
\end{pmatrix} ,
$$
where the $(p+1)d\times (p+1)d$ lower triangular matrix $\LL$ is also lower 
triangular
with respect to its blocks which are partly scalars, partly vectors, 
partly matrices as follows: 
$$
 {\ell}_{ij} = \left\{\begin{array}{ll}
  a_{ji},  & j=1,\dots ,d-1; \quad i=j+1, \dots ,d;  \\
  1 & i=j=1,\dots ,d  ;  \\
  0 & i=1,\dots ,d; \quad j=i+1,\dots ,(p+1) d;
                 \end{array}
                 \right.
$$
further, the vectors $\bm{\ell}_{d+1,j}$ are $pd\times 1$ for $j=1,\dots ,d$, 
and 
comprise the column vectors of the $pd\times d$ matrix $\B^T$. The matrix
in the bottom right block is the $pd\times pd$ identity, and above it, the
zero entries can be arranged into the $d\times pd$ zero matrix.

The $(p+1)d\times (p+1)d$ block-diagonal matrix $\DD$ in partitioned form is
$$
\DD = \begin{pmatrix}
  \delta_1^{-1}      &0       &0      &\dots  &0 &0 &\dots &0 \\
  0             &\delta_2^{-1}       &0       &\dots  &0 &0 &\dots &0 \\
  0             &0   &\delta_3^{-1}      &\dots  &0 &0 &\dots &0 \\
 \vdots  &\vdots  &\vdots &\vdots  &\vdots &\0 &\dots &\0 \\
 0  &0 &0 &0  &\delta_d^{-1} &0 &\dots &0 \\
 \0  &\0 &\vdots &\0 
 &\0 & & \mathfrak{C}_p^{-1}  & 
\end{pmatrix} ,
$$
where the $pd\times 1$ vectors $\0$ comprise $\OO_{pd \times d}$ in the 
left bottom, and  the matrix $\mathfrak{C}_p^{-1}$  stands in the right
bottom block. With multiplication rules
of block matrices and their inverses,
the recursion of the block LDL decomposition goes on as follows.
\begin{itemize}
\item Outer cycle (column-wise). For $j=1,\dots ,d$: 
 $\delta^{-1}_j = k_{jj} -\sum_{h=1}^{j-1} {\ell}_{jh} \delta_h^{-1} {\ell}_{jh}$
 (with the reservation that $\delta^{-1}_1 = k_{11}$);
\item Inner cycle (row-wise). For $i=j+1 ,\dots ,d$:
$$ {\ell}_{ij} = \left( k_{ij} - \sum_{h=1}^{j-1} {\ell}_{ih} \delta_h^{-1}  
 {\ell}_{jh} \right) \delta_j 
$$
and   
$$
 \bm{\ell}_{d+1,j} = \left( \bm{k}_{d+1,j} - \sum_{h=1}^{j-1} 
 \bm{\ell}_{d+1,h} \delta_h^{-1}  {\ell}_{jh} \right) \delta_j  
$$
(with the reservation that in the $j=1$ case the summand is zero),
where  $\bm{k}_{d+1,j}$ for $j=1,\dots ,d$ are $pd\times 1$ vectors 
in the bottom left block of $\K$. 
\end{itemize}
The recursion ends in one run.

The above decomposition is again a nested one, so
for the first $d$ rows of $\LL$, 
only its previous rows or preceding entries in the same row enter into
the calculation, as if we performed the usual LDL decomposition of $\K$.
 Therefore,
${\ell}_{ij} = a_{ji}$ for $j=1,\dots ,d-1$, $i=j+1, \dots ,d$ that are
the negatives of the partial regression coefficients 
 akin to those offered by the standard LDL decomposition 
$\K = {\tilde \LL } {\tilde \DD } {\tilde \LL }^T$; so the first $d$ rows of
$\tilde \LL$ and $\LL$ are the same, and the first $d$ rows of
$\tilde \DD$ and $\DD$ are the same too.

When the process terminates, we consider the blocks ``en block'' and get the
$pd\times d$ matrix
$\B^T =( \bm{\ell}_{d+1,1}, \dots ,  \bm{\ell}_{d+1,d})$.

\section{Graphical Models}\label{appB}
   
Here we give a short description of graphical models, 
with emphases to the Gaussian case, based on~\cite{Bollacta,Lauritzen}.
Directed and undirected models  have many properties in common, 
and under some conditions, there are important correspondences between them.
Graphical models represent conditional independence  properties
of the components of a random vector by means of the structural properties 
of a suitably constructed graph.

\subsection{Directed Graphical Model: Bayesian Network}\label{bn}

The nodes of the graph $G$ correspond to random variables  
$X_1 ,\dots ,X_d$, 
whereas the directed edges to \textit{causal} 
dependences between them.  In case of a DAG  $G$ 
with node-set $V=\{ 1,\dots ,d \}$,
there are no directed cycles, and therefore,
there exists a recursive ordering (labeling) of
the nodes such that for every directed edge $j\to i$, $i<j$ holds;
we can refer to this relation as $j$ is the \textit{parent} of $i$.
So the youngest node has label 1, and the older
a node, the larger its label is (we can think of labels as ages).
We use this, so-called (not necessarily unique) 
\textit{topological labeling} of the node-set.
(Note that some authors use the opposite ordering.)

There can be multiple parents of $i$ (maximum $d-i$ ones).
Let $\pa (i) \subset \{ {i+1} ,\dots ,d \}$ denote the set of the parents 
of $i$, and   for any $A\subset V$ we use the notation
$\x_A = \{ x_i : \, i\in A \}$ and $\X_A = \{ X_i : \, i\in A \}$.
To draw the edges, the
\textit{directed pairwise Markov property}  is used: for $i<j$,
there is no $j\to i$
directed edge, whenever  $X_i$ and $X_{j}$ are conditionally independent, 
given $\X_{\pa (i)}$. With notation,
$$
 X_i \indep X_{j} | \X_{\pa (i)} \quad \textrm{for} \quad j\in 
 \{ {i+1} ,\dots ,d \} \setminus \pa (i),  \quad i=1,\dots ,d-1 .
$$

\subsection{Undirected Graphical Model: Markov Random Field (MRF)}\label{mrf}

The nodes of the graph $G$ also correspond to random variables 
$X_1 ,\dots ,X_d$, but  the conditional independence 
statements here include the neighbors of the nodes.
The \textit{undirected global Markov property} 
of a joint distribution with respect to  $G$ is defined as follows:
$$
 \X_A \indep \X_B | \X_S   
$$
holds for any node cutset $S$ between disjoint 
node subsets $A$ and $B$ (i.e. removing nodes of $S$ will make $A$ and $B$
disconnected).
The \textit{undirected pairwise Markov property}  is defined as 
follows: nodes $i$ and $j$ are not connected whenever
$$
 X_i \indep X_j | \X_{V\setminus \{ i,j \} } , \quad i\neq j 
$$
holds.
The \textit{undirected factorization property}  means the
factorization of the joint density, for any state configuration
$\x = (x_1 ,\dots ,x_d )$ as
$$
 f ( \x ) = \frac1{Z}\prod_{C \in {\CC }} \Psi_{C} (\x_C )  
$$
with normalizing constant $Z>0 $ and 
non-negative \textit{compatibility function}s  $\Psi_C$s assigned to 
the cliques  of $G$. 
Under \textit{clique} we understand a maximal complete subgraph.
(Note that, in graph theory, they are sometimes called maximal cliques.)
 The above factorization  is far not unique, but in special 
(so-called called decomposable) models, the
forthcoming Equation~\eqref{marko} gives an explicit formula for the 
compatibility functions.
The Hammersley--Clifford theorem (see~\cite{Lauritzen}) states that in case
of positive distributions, the above undirected Markov 
properties are equivalent.
(Positivity means that the joint density is absolutely continuous with
respect to the direct product of its marginal distributions, e.g. in case of a 
non-degenerate multivariate Gaussian distribution.)

Also, even if the underlying graph is undirected, a decomposable structure of it
(see Subsection~\ref{decomp})
gives a (not necessarily unique) so-called \textit{perfect ordering} 
of the nodes, in which order directed edges can be drawn. 
Conversely, a directed graph can be made undirected
by disregarding the orientation
of the edges  and introducing possible new undirected edges by ``moralization'':
we connect two parents (having a common child) whenever they are
not connected (married). The so obtained \textit{moral graph}
 is then used in the MRF setup.
Moralization is needed when for some triplet $i,j>k$ in the DAG $G$, 
$i\to k$ and $j\to k$ holds, but there is no directed edge between $i$ and 
$j$. Such a triplet is called \textit{sink V}, see Figure~\ref{FATMA}.
\begin{figure}
	\tikzstyle{vertex} = [draw, fill, circle, node distance=40pt, ultra thick, scale=0.1]
	\tikzstyle{dummy} = [node distance=20pt]
	\tikzstyle{edge} = [opacity=1, cap=round, join=round]
	\tikzstyle{redge} = [opacity=1, -latex,  thick]
\begin{center}
		\begin{tikzpicture}

		\node[vertex,  label =above:{$i=$2}] (a) {j};	
	\node[vertex, node distance = 400pt, left of = a, label =above:{$j=$3}] (b) {2};	
		\node[vertex, node distance = 300pt, below right of = b, label =below:{$k=$1}] (c) {1};
		
		\draw[redge] (a) -- (c) ;
		\draw[redge] (b) -- (c) ;
		\end{tikzpicture}
	\end{center}
	\caption{Triplet \textit{sink V}.}\label{FATMA}
\end{figure}
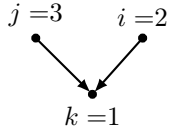
 Then even if
$X_i$ and $X_j$ are (marginally) independent, they are not conditionally
independent any more, given $X_k$. For example, 
if $X_i$ is the years of former schooling and $X_j$ is the gender, then 
-- though they are independent (men and women can get any education irrespective
of gender) -- 
they are conditionally dependent given the income ($X_k$). Indeed,
in the example of~\cite{Wer15}, on given level of salary, women had a higher 
level of education  than men. 

The original graph is Markov equivalent to its moral graph 
(see~\cite{Lauritzen}, Proposition 3.28), 
i.e. the directed and undirected Markov properties
hold for them at the same time.

\subsection{Decomposable Models}\label{decomp}

The definition of the (weak) decomposability of a graph is recursive.
\begin{Definition}\label{weak} 
The graph $G$ is \textit{decomposable}  if it is either a complete graph or 
its node-set $V$ can be partitioned into disjoint node subsets 
$A,B,C$  such that 
\begin{itemize}
\item 
$C$ defines a complete subgraph;
\item
$C$ separates $A$ from $B$
(in other words, $C$ is a node cutset between $A$ and $B$);
\item
the subgraphs generated by $A\cup C$ and $B\cup C$ are both decomposable.
\end{itemize} 
\end{Definition}
Here we establish many equivalent properties of a decomposable graph,
based on~\cite{Bollacta,Lauritzen,Wermuth}:
\begin{itemize}
\item  $G$ is 
\textit{triangulated} (with other words, \textit{chordal}), 
i.e. every cycle in $G$ of length at least four has a chord.

\item
$G$ has a \textit{perfect numbering} of its nodes such that in this
labeling, $\textrm{ne} (i) \cap \{ i+1 ,\dots ,d \}$
is a complete subgraph, where $\textrm{ne} (i)$ is the set of neighbors of $i$,
for $i=1, \dots ,d $. It is also called \textit{single
node elimination ordering}
(see~\cite{Wainwright}), and obtainable with the
\textit{Maximal Cardinality Search (MCS)} algorithm of~\cite{Tarjan}.

\item
$G$ has the following \textit{running intersection property}:  
we can number the  cliques of it to form a so-called
\textit{perfect sequence} $C_1 ,\dots ,C_k$ where each combination of 
the subgraphs
induced by $H_{j-1} = C_1 \cup \dots \cup C_{j-1}$ and $C_j$ is a decomposition
$(j=2,\dots ,k)$, i.e. the necessarily complete subgraph 
$S_j = H_{j-1} \cap C_j$ is a separator. More precisely,  $S_j$ is a node 
cutset between the disjoint node subsets $H_{j-1} \setminus S_j$ and 
$R_j =C_j \setminus S_j = H_j \setminus H_{j-1}$. 
This sequence of cliques is also  called a \textit{junction tree (JT)}.

Here any clique $C_j$ is the disjoint union of $R_j$ (called \textit{residual}),
 the nodes of which
are not contained in any $C_i$, $i<j$ and of $S_j$ (called \textit{separator})
with the following  property: there is an  $i^*\in \{ 1, \dots ,j-1 \}$ 
such that 
$$
  S_j = C_j \cap (\cup_{i=1}^{j-1} C_i )  = C_j \cap C_{i^*} .
$$
This (not necessarily unique) $C_{i^*}$ is called 
\textit{parent clique} of $C_j$.
Here $S_1 =\emptyset$ and $R_1 =C_1$.
Furthermore, if such an ordering is possible, a version may be found in which
any prescribed set is the first one.  
Note that the junction tree is indeed a tree with nodes $C_1 ,\dots ,C_k$
and one less edges, that are the separators $S_2 ,\dots ,S_{k}$.
 
\item There is a labeling of the nodes such that the adjacency 
matrix contains a \textit{reducible zero pattern (RZP)}. It means that
the zero entries in the upper-diagonal part of the adjacency matrix form
an index set that is reducible in the following sense.
The index set $I$, which is the subset of the set of edges
$\{ (i,j): \, 1\le i <j \le d\}$, is called reducible if for each 
$(i,j)\in I$ and $h=1,\dots ,i-1$, we have $(h,i) \in I$ or 
$(h,j) \in I$ or both.
 
Indeed, this convenient labeling is a perfect numbering
of the nodes.


\item
The following \textit{Markov chain} property also holds:
$f (\x_{R_j} \, | \, \x_{C_{1} \cup \dots \cup C_{j-1} } )=
  f (\x_{R_j} \, | \, \x_{S_j} )$.

Therefore, if we have a perfect sequence $C_1 ,\dots ,C_k$
of the cliques with separators $S_1 =\emptyset , S_2 ,\dots ,S_k$, 
then for any state configuration $\x$ we have the following form of the density:
\begin{equation}\label{marko}
 f (\x ) = \frac{\prod_{j=1}^k f (\x_{C_j} )}{\prod_{j=2}^k f (\x_{S_j} )} = 
\prod_{i=1}^k f (\x_{R_j} |\x_{S_j} ) .
\quad 
\end{equation}
\end{itemize}

To find the structure, where one of the equivalent criteria of decomposability 
holds, 
we can use the  MCS method of~\cite{Tarjan}.
The simple MCS gives   label $d$ to an arbitrary  node. 
Then the nodes are labeled consecutively, from $d$ down to 1, 
choosing as the next to label a node with a
maximum number of previously labeled neighbors and breaking ties arbitrarily.
(Note that~\cite{Lauritzen} labels the nodes conversely.) 
The MCS ordering is far not unique, and this simple version is not always
capable to find the JT structure behind a triangulated graph in one run,
but another run is needed. There are also variants of this algorithm which
are applicable to a non-triangulated graph too, and capable to triangulate it
with adding minimum number of edges.


For example, let us consider the following adjacency matrices on four nodes:
$$
 \begin{pmatrix}
   1  & 1 & 1 & 1 \\
   1  & 1 & 0 & 1 \\
   1  & 0 & 1 & 1 \\
   1  & 1 & 1 & 1 
 \end{pmatrix} , \quad
\begin{pmatrix}
   1  & 1 & 0 & 1 \\
   1  & 1 & 1 & 1 \\
   0  & 1 & 1 & 1 \\
   1  & 1 & 1 & 1 
 \end{pmatrix} , \quad
\begin{pmatrix}
   1  & 0 & 1 & 1 \\
   0  & 1 & 0 & 1 \\
   1  & 0 & 1 & 1 \\
   1  & 1 & 1 & 1 
 \end{pmatrix} , \quad
\begin{pmatrix}
   1  & 1 & 0 & 1 \\
   1  & 1 & 1 & 0 \\
   0  & 1 & 1 & 1 \\
   1  & 0 & 1 & 1 
 \end{pmatrix} .
$$
Then the upper diagonal part of the first adjacency matrix  does not have the 
RZP, due to the sink V configuration $2 \rightarrow 1 \leftarrow 3$. 
However, as the skeleton is
triangulated, there must be a relabeling of the nodes in which it has the
RZP; e.g. in the 2,1,3,4 permutation of the nodes, where in the second 
matrix,
there is no sink V configuration in the corresponding directed graph.
The third adjacency matrix has an RZP; equivalently, in this
ordering of the variables,  there is no sink V configuration in the
directed graph; further, the skeleton is indeed triangulated.
The last adjacency matrix does not have an RZP in any ordering of the
variables; equivalently, the graph is not triangulated. 
Therefore, we cannot construct a Markov equivalent DAG to it.

\begin{Corollary}\label{moral}
In case of a decomposable model, the perfect ordering of the nodes
(in which the adjacency matrix of the undirected graph has the RZP)
defines a DAG that does not have a sink V configuration. Therefore, it
is Markov equivalent to its undirected skeleton with the same 
conditional independence statements for the components of the underlying
multivariate distribution. Consequently, for $i<j$, there is no directed
$j\to i$ edge if and only if $i$ and $j$ are not connected in the
undirected graph.
\end{Corollary}

\subsection{Undirected Gaussian Graphical Model}\label{gauss}

Let $\X \sim {\cal N}_d (\muv ,\Siga )$ be a $d$-variate 
non-degenerate Gaussian random vector 
with expectation $\muv $ and
positive definite, symmetric $d\times d$ covariance matrix $\Siga$. 
The also
positive definite, symmetric matrix $\Siga^{-1}$ of entries $\sigma^{ij}$ 
is called  \textit{concentration matrix}, it appears in the joint density
and its zero entries indicate conditional independences between
two components of $\X$, given the remaining ones.
Mostly, the variables are already centered, so $\muv =\0$ is assumed.

Let us form an undirected graph $G$ on the node-set $V$, 
where $V$ corresponds to the
components of $\X$ and the edges are drawn according to the rule
$$
 i\sim j  \Leftrightarrow \sigma^{ij} \neq 0, \quad i\neq j .
$$ 
This is called undirected \textit{Gaussian graphical model}.
To establish conditional independence statements, we use the following facts.

\begin{Proposition}\label{partial}
Let $\X =(X_1 ,\dots ,X_d )^T \sim {\cal N}_{d} (\0 ,\Siga )$ be a 
random vector, 
and let $V :=
\{ 1,\dots ,d \}$ denote the index set of the variables, $d\ge 3$. Assume
that $\Siga$ is positive definite. Then
$$
 r_{X_i X_j |\X_{V \setminus \{ i, j \} }} = 
    \frac{-\sigma^{ij}}{ \sqrt{ \sigma^{ii} \sigma^{jj} } }  \qquad i\neq j ,
$$
where $r_{X_i X_j |\X_{V \setminus \{ i, j \} }}$ denotes the partial correlation
coefficient between $X_i$ and $X_j$ after eliminating the effect of the
remaining variables $\X_{V \setminus \{ i, j \} }$. Further,
$$
 \sigma^{ii} = 1/(\Var (X_i | \X_{V \setminus \{ i \} } ) , \qquad
   i=1,\dots ,d
$$
is the reciprocal of the conditional (residual) variance of 
$X_i$, given the other variables $\X_{V \setminus \{ i \} }$.
\end{Proposition}

\begin{Definition}\label{reg}
Let $\X \sim {\cal N}_{d} (\0 , \Siga )$ be random vector 
with $\Siga$ positive definite.  Consider the regression plane
$$
 \E (X_i | \X_{V \setminus \{ i \} } =\x_{V \setminus \{ i \} } ) 
 =\sum_{j \in V \setminus \{ i \} } 
 \beta_{ji \cdot V \setminus \{ i \} } x_j , \quad j\in V \setminus \{ i \},
$$
where $x_j$'s are the coordinates of $\x_{V \setminus \{ i \} }$. Then
we call the  coefficient $\beta_{ji \cdot V \setminus \{ i \} }$ the
\textit{partial regression coefficient} of $X_j$ when regressing $X_i$
with $\X_{V \setminus \{ i \} }$, $j \in  V \setminus \{ i \}$. 
\end{Definition}

\begin{Proposition}\label{partregr}
$$
 \beta_{ji \cdot V \setminus \{ i \} } =  -\frac{\sigma^{ij}}{\sigma^{ii}} ,
  \quad j\in V \setminus \{ i \}.
$$
\end{Proposition}

\begin{Corollary}\label{cor}
An important consequence of Propositions~\ref{partial} and~\ref{partregr} is that 
$$
 \beta_{ji \cdot V \setminus \{ i \} } = 
 r_{X_i X_j |\X_{V \setminus \{ i, j \} }} \sqrt{\frac{ \sigma^{jj}}{ \sigma^{ii} } } =
r_{X_i X_j |\X_{V \setminus \{ i, j \} }} 
 \sqrt{\frac{\Var (X_i | \X_{V \setminus \{ i \} } ) }{ \Var (X_j | 
 \X_{V \setminus \{ j \} } ) } } ,  \quad j\in V \setminus \{ i \}.
$$
\end{Corollary}
(The formula is analogous to the one of unconditioned regression.)
So only  the variables $X_j$'s whose partial correlation with $X_i$
(after eliminating the effect of the remaining variables) is not 0, enter into
the regression of $X_i$ with the other variables. 

For $i\neq j$ we want to test
$$
 H_0 \, : \, r_{X_i X_j |\X_{V \setminus \{ i, j \} }} =0,
$$
i.e. that $X_i$ and $X_j$ are conditionally independent conditioned
on the remaining variables. Equivalently, $H_0$ means that 
$\beta_{ij | V \setminus \{ i \} } =0$,
$\beta_{ji | V \setminus \{ j \} } =0$, or simply, $\sigma^{ij}=\sigma^{ji}=0$
($\Siga >0$ is assumed).

To test $H_0$ in some form, several exact tests are known that are
usually based on likelihood  ratio tests. 
The following test uses the empirical partial
correlation coefficient, denoted by
${\hat r}_{X_i X_j |\X_{V  \setminus \{ i, j \} }}$, and the following
statistic is based on it:
$$
 B=1- ( {\hat r}_{X_i X_j |\X_{V \setminus \{ i, j \} }} )^2 =
 \frac{|\SSS_{V \setminus \{ i,j \} }| \cdot |\SSS_{V} | }
  {|\SSS_{V \setminus \{ i\} }| \cdot |\SSS_{V \setminus \{ j\} } |  } ,
$$  
where $\SSS$ is the sample size ($n$) times the empirical covariance matrix
of the variables in the subscript
(its entries are the product-moments).

It can be proven that under $H_0$, the test statistic
$$
 t=\sqrt{n-d} \cdot \sqrt{ \frac1{B} -1 } 
 =\sqrt{n-d} \cdot \frac{{\hat r}_{X_i X_j |\X_{\V \setminus \{ i, j \} }} }
 {\sqrt{1-({\hat r}_{X_i X_j |\X_{V \setminus \{ i, j \} } } )^2  } }
$$
is distributed as Student's $t$ with $n-d$ degrees of freedom. Therefore,
we reject $H_0$ for large values of $|t|$.

\subsection{Covariance Selection}\label{undirg}

For practical purposes we use the empirical partial correlation coefficients,
and based on them,  the above exact test to check whether they
significantly differ from 0 or not. 
  In this way, we put
zeros into the no-edge positions $ij$'s of the inverse covariance 
matrix (where partial correlations do not significantly differ from zero), 
and fit a so-called \textit{covariance selection model}. 
The restricted covariance matrix is denoted by $\Siga^*$.
We want to estimate $\Siga^*$ from the iid.
sample $\X_1 ,\dots ,\X_n \sim {\cal N}_{d} (\0, \Siga )$ $(n>d)$, 
such that the estimate ${\hat \K}$ of  ${\Siga^*}^{-1}$ 
has zero entries in the no-edge  positions. 

In Theorem 5.3 of~\cite{Lauritzen}) it is proved, that under the covariance
selection model,  the ML-estimator of
the mean vector is the sample mean $\bar \X$, and the restricted 
covariance matrix $\Siga^* =(\sigma^*_{ij})$ 
is estimated as follows.  The entries in the
edge-positions are estimated as in the saturated model (no restrictions): 
\begin{equation}\label{ips}
 {\hat \sigma^* }_{ij} = \frac1{n} s_{ij} , \quad \{ i,j \} \in E ,
\end{equation}
where $\SSS =(s_{ij} ) =\sum_{\ell =1}^{n} ( \X_{\ell} -{\bar \X} ) ( \X_{\ell}  -
{\bar \X} )^T$ is the usual product moment estimate.
The other entries (in the no-edge positions) of $\Siga^*$ are free, 
but satisfy the model conditions: after taking the inverse $\K$ of 
$\Siga^*$ with these undetermined entries,  we get the same number 
of equations for them from $k^{ij} =0$ whenever $\{ i,j \} \notin E$. 
To do so, there are numerical algorithms at our disposal, for instance, the 
 Iterative Proportional Scaling (IPS), see~\cite{Lauritzen}, p. 134, where
an infinite iteration is needed, because in general, there is no explicit 
solution for the ML estimate. However, the fixed point of this iteration
gives a unique positive definite matrix $\hat \K$. 

In the decomposable case, there is no need of
running the IPS, but explicit estimate can be given as follows.
Recall that if the Gaussian graphical model is decomposable 
(its concentration graph $G$ is decomposable),
then the cliques,
together with their separators (with possible multiplicities), form a JT
structure. Denote $\cal C$ the set of the 
cliques and $\cal S$ the set of the separators in $G$.
Then direct density estimates, using (\ref{marko}), are available.
In particular, the ML estimator of $\K$ can be calculated based on the 
product moment
estimators applied for subsets of the variables, corresponding
to the cliques and separators.

Let $n$ be the size of the sample for the underlying $d$-variate normal 
distribution, and assume that $n>d$.
For the clique $C \in {\cal C}$, let  $[\SSS_{C}]^{V }$ denote
$n$ times the empirical covariance matrix corresponding to the variables
$\{ X_i : \, i\in C \}$ complemented with zero entries to have a
$d\times d$ (symmetric, positive semidefinite) matrix. Likewise,
for the separator $S \in {\cal S}$, let  $[\SSS_{S}]^{V }$ denote
$n$ times the empirical covariance matrix corresponding to the variables
$\{ X_i : \, i\in S \}$ complemented with zero entries to have an 
$d\times d$ (symmetric, positive semidefinite) matrix. Then the ML
estimator of the mean vector is the sample average (as usual), while the
ML estimator of the concentration matrix is
\begin{equation}\label{Kestimated}
 {\hat \K} = n\left\{ \sum_{C\in {\cal C}} [\SSS_{C}^{-1}]^{V } -
             \sum_{S\in {\cal S}} [\SSS_{S}^{-1}]^{V }  \right\} ,
\end{equation}
see Proposition 5.9 of~\cite{Lauritzen}. This proposition states that 
the above ML estimate exists with probability one if and only if 
$n$ is greater than the maximum clique size.
In the financial example of Section~\ref{appl}, the cliques and separators of
Equation~\eqref{klikkek} are used. Then the estimate of $\K$ is as follows:
$$
\begin{aligned}
\hat \K &=  
\begin{pmatrix}
0 & 0 & 0 & 0 & 0 & 0 & 0 & 0 \\
0 & 0 & 0 & 0 & 0 & 0 & 0 & 0 \\
0 & 0 & \hat{k}_{\{3,4,5,6,7,8\}}^{33} & \hat{k}_{\{3,4,5,6,7,8\}}^{34} & \hat{k}_{\{3,4,5,6,7,8\}}^{35} & \hat{k}_{\{3,4,5,6,7,8\}}^{36} & \hat{k}_{\{3,4,5,6,7,8\}}^{37} & \hat{k}_{\{3,4,5,6,7,8\}}^{38} \\
0 & 0 & \hat{k}_{\{3,4,5,6,7,8\}}^{43} & \hat{k}_{\{3,4,5,6,7,8\}}^{44} & \hat{k}_{\{3,4,5,6,7,8\}}^{45} & \hat{k}_{\{3,4,5,6,7,8\}}^{46} & \hat{k}_{\{3,4,5,6,7,8\}}^{47} & \hat{k}_{\{3,4,5,6,7,8\}}^{48} \\
0 & 0 & \hat{k}_{\{3,4,5,6,7,8\}}^{53} & \hat{k}_{\{3,4,5,6,7,8\}}^{54} & \hat{k}_{\{3,4,5,6,7,8\}}^{55} & \hat{k}_{\{3,4,5,6,7,8\}}^{56} & \hat{k}_{\{3,4,5,6,7,8\}}^{57} & \hat{k}_{\{3,4,5,6,7,8\}}^{58} \\
0 & 0 & \hat{k}_{\{3,4,5,6,7,8\}}^{63} & \hat{k}_{\{3,4,5,6,7,8\}}^{64} & \hat{k}_{\{3,4,5,6,7,8\}}^{65} & \hat{k}_{\{3,4,5,6,7,8\}}^{66} & \hat{k}_{\{3,4,5,6,7,8\}}^{67} & \hat{k}_{\{3,4,5,6,7,8\}}^{68} \\
0 & 0 & \hat{k}_{\{3,4,5,6,7,8\}}^{73} & \hat{k}_{\{3,4,5,6,7,8\}}^{74} & \hat{k}_{\{3,4,5,6,7,8\}}^{75} & \hat{k}_{\{3,4,5,6,7,8\}}^{76} & \hat{k}_{\{3,4,5,6,7,8\}}^{77} & \hat{k}_{\{3,4,5,6,7,8\}}^{78} \\
0 & 0 & \hat{k}_{\{3,4,5,6,7,8\}}^{83} & \hat{k}_{\{3,4,5,6,7,8\}}^{84} & \hat{k}_{\{3,4,5,6,7,8\}}^{85} & \hat{k}_{\{3,4,5,6,7,8\}}^{86} & \hat{k}_{\{3,4,5,6,7,8\}}^{87} & \hat{k}_{\{3,4,5,6,7,8\}}^{88}
\end{pmatrix} \\
&+
\begin{pmatrix}
0 & 0 & 0 & 0 & 0 & 0 & 0 & 0 \\
0 & \hat{k}_{\{2,3,5,6,7\}}^{22} & \hat{k}_{\{2,3,5,6,7\}}^{23} & 0 & \hat{k}_{\{2,3,5,6,7\}}^{25} & \hat{k}_{\{2,3,5,6,7\}}^{26} & \hat{k}_{\{2,3,5,6,7\}}^{27} & 0 \\
0 & \hat{k}_{\{2,3,5,6,7\}}^{32}& \hat{k}_{\{2,3,5,6,7\}}^{33} & 0 & \hat{k}_{\{2,3,5,6,7\}}^{35} & \hat{k}_{\{2,3,5,6,7\}}^{36} & \hat{k}_{\{2,3,5,6,7\}}^{37} & 0 \\
0 & 0 & 0 & 0 & 0 & 0 & 0 & 0 \\
0 & \hat{k}_{\{2,3,5,6,7\}}^{52} & \hat{k}_{\{2,3,5,6,7\}}^{53} & 0 & \hat{k}_{\{2,3,5,6,7\}}^{55} & \hat{k}_{\{2,3,5,6,7\}}^{56} & \hat{k}_{\{2,3,5,6,7\}}^{57} & 0 \\
0 & \hat{k}_{\{2,3,5,6,7\}}^{62} & \hat{k}_{\{2,3,5,6,7\}}^{63} & 0 &\hat{k}_{\{2,3,5,6,7\}}^{65} & \hat{k}_{\{2,3,5,6,7\}}^{66} & \hat{k}_{\{2,3,5,6,7\}}^{67} & 0 \\
0 & \hat{k}_{\{2,3,5,6,7\}}^{72} & \hat{k}_{\{2,3,5,6,7\}}^{73} & 0 & \hat{k}_{\{2,3,5,6,7\}}^{75} & \hat{k}_{\{2,3,5,6,7\}}^{76} & \hat{k}_{\{2,3,5,6,7\}}^{77} & 0 \\
0 & 0 & 0 & 0 & 0 & 0 & 0 & 0 \\
\end{pmatrix} \\
&+
\begin{pmatrix}
\hat{k}_{\{1,4,5\}}^{11} & 0 & 0 & \hat{k}_{\{1,4,5\}}^{14} & \hat{k}_{\{1,4,5\}}^{15} & 0 & 0 & 0 \\
0 & 0 & 0 & 0 & 0 & 0 & 0 & 0 \\
0 & 0 & 0 & 0 & 0 & 0 & 0 & 0 \\
\hat{k}_{\{1,4,5\}}^{41} & 0 & 0 & \hat{k}_{\{1,4,5\}}^{44} & \hat{k}_{\{1,4,5\}}^{45} & 0 & 0 & 0 \\
\hat{k}_{\{1,4,5\}}^{51} & 0 & 0 & \hat{k}_{\{1,4,5\}}^{54} & \hat{k}_{\{1,4,5\}}^{55} & 0 & 0 & 0 \\
0 & 0 & 0 & 0 & 0 & 0 & 0 & 0 \\
0 & 0 & 0 & 0 & 0 & 0 & 0 & 0 \\
0 & 0 & 0 & 0 & 0 & 0 & 0 & 0 \\
\end{pmatrix} \\
&-
\begin{pmatrix}
0 & 0 & 0 & 0 & 0 & 0 & 0 & 0 \\
0 & 0 & 0 & 0 & 0 & 0 & 0 & 0 \\
0 & 0 & \hat{k}_{\{3,5,6,7\}}^{33} & 0 & \hat{k}_{\{3,5,6,7\}}^{35} & \hat{k}_{\{3,5,6,7\}}^{36} & \hat{k}_{\{3,5,6,7\}}^{37} & 0 \\
0 & 0 & 0 & 0 & 0 & 0 & 0 & 0 \\
0 & 0 & \hat{k}_{\{3,5,6,7\}}^{53} & 0 & \hat{k}_{\{3,5,6,7\}}^{55} & \hat{k}_{\{3,5,6,7\}}^{56} & \hat{k}_{\{3,5,6,7\}}^{57} & 0 \\
0 & 0 & \hat{k}_{\{3,5,6,7\}}^{63} & 0 & \hat{k}_{\{3,5,6,7\}}^{65} & \hat{k}_{\{3,5,6,7\}}^{66} & \hat{k}_{\{3,5,6,7\}}^{67} & 0 \\
0 & 0 & \hat{k}_{\{3,5,6,7\}}^{73} & 0 & \hat{k}_{\{3,5,6,7\}}^{75} & \hat{k}_{\{3,5,6,7\}}^{76} & \hat{k}_{\{3,5,6,7\}}^{77} & 0 \\
0 & 0 & 0 & 0 & 0 & 0 & 0 & 0 \\
\end{pmatrix} \\
&-
\begin{pmatrix}
0 & 0 & 0 & 0 & 0 & 0 & 0 & 0 \\
0 & 0 & 0 & 0 & 0 & 0 & 0 & 0 \\
0 & 0 & 0 & 0 & 0 & 0 & 0 & 0 \\
0 & 0 & 0 & \hat{k}_{\{4,5\}}^{44} & \hat{k}_{\{4,5\}}^{45} & 0 & 0 & 0 \\
0 & 0 & 0 & \hat{k}_{\{4,5\}}^{54} & \hat{k}_{\{4,5\}}^{55} & 0 & 0 & 0 \\
0 & 0 & 0 & 0 & 0 & 0 & 0 & 0 \\
0 & 0 & 0 & 0 & 0 & 0 & 0 & 0 \\
0 & 0 & 0 & 0 & 0 & 0 & 0 & 0 \\
\end{pmatrix} \\
\end{aligned}
$$

\subsection{Directed Gaussian Graphical Model}\label{rec}

Here, in a topological (DAG) ordering of the variables, 
recursive linear regressions are
introduced in such a way that $X_i$ is regressed with $X_j$s for $j>i$ and
a zero partial correlation coefficient means no $j\to i$ directed edge.
 In~\cite{Wermuth} the following model of recursive
linear equations was introduced.
Let $\X \sim {\cal N}_d (\0 ,\Siga )$ be a $d$-dimensional Gaussian random 
vector with real coordinates of zero expectation. Then
\begin{equation}\label{matrix}
 \A \X =\U \quad \textrm{with} \quad \U =(U_1 ,\dots ,U_d )^T  
\sim {\cal N}_d (\0 ,\Dada ) ,
\end{equation}  
where $\A$ is a $d\times d$  upper triangular 
matrix with 1s along its
main diagonal, otherwise it contains the negatives of the partial regression 
coefficients $a_{ij}$'s, when $X_i$ is the target of a multivariate 
linear regression with predictors $\{ X_j : \, j >i \}$. 
Actually, $a_{ij}$ is a ``path coefficient'' that is a scaled measure of the 
$X_j \to X_i$ influence, and
one can use statistical tests for its significance.
As for $\Dada =\diag (\delta_1 ,\dots ,\delta_d )$, it
is a $d\times d$ diagonal matrix with positive diagonal entries, 
covariance matrix of the error term $\U$.

Taking the covariance matrix on both sides of~\eqref{matrix}, we get
\begin{equation}\label{kell}
 \E [ ( \A \X ) (\A \X )^T ] =\A \Siga \A^T = \Dada .
\end{equation}
In~\cite{Wermuth}, the entries of $\A$ and $\Dada$ are characterized as partial
correlations and residual variances. However, in~\cite{Bollacta} it is shown
that the decomposition in Equation~\eqref{kell} is obtainable by the 
standard LDL (variant of the simple Cholesky)  
decomposition of  $\Siga^{-1}$ as follows:
\begin{equation}\label{LDLdecomp}
 \Siga^{-1} = \LL \Dada^{-1} \LL^T .
\end{equation}
This decomposition of the positive definite matrix
$\Siga^{-1}$ is unique, where $\LL $ is lower triangular of entries $1$s 
along its main diagonal and $\Dada^{-1}$ is a diagonal matrix of 
entries all positive along its main diagonal. 
By uniqueness, $\A =\LL^T$ and $\Dada$ give the 
solution for the parameters of the original model (\ref{matrix}).

Without using the LDL decomposition,  in~\cite{Wermuth} the
following explanation for $a_{ij}$ is  given based on the 
recursive equations
\begin{equation}\label{wer}
 X_i = -\sum_{j=i+1}^d a_{ij} X_j + U_i , \quad i=1,\dots ,d.
\end{equation}
Equation~\eqref{wer}  shows that for $i<j$, 
 $a_{ij} =-\beta_{ij \cdot \{ i+1 ,\dots ,d \} }$, where
$\beta_{ij \cdot \{ i+1 ,\dots ,d \} }$ is  the  partial regression 
coefficient of $X_j$ when  regressing $X_i$
with $X_{i+1} ,\dots ,X_d $. 
So $a_{ij} =0$ if and only if $X_i$ and $X_j$ are conditionally independent, 
given
the variables with indices in the conditioning set (this means conditional and
not marginal independence). However, to find $a_{ij}$ we do not need this
formula, but we can find all upper-diagonal entries of $\A$, at the same time,
 with the LDL decomposition of the concentration matrix in~\eqref{LDLdecomp}. 
Actually, this is the 
same as the algorithm proposed in~\cite{Wermuth}, 
when $a_{ij}$ is obtained with the inverse of the 
restricted covariance matrix marginalized 
to $X_i , X_{i+1} ,\dots ,X_d$; it uses the sweeping technique of the 
Cholesky recursion  for each column of $\A$ separately.

Note that at this point, the ordering of the jointly 
Gaussian variables is not relevant, since in any recursive ordering of them 
(encoded in $\A$) a Gaussian directed graphical model
(in other words, a Gaussian Bayesian network) can be constructed, where every 
variable is regressed linearly with the higher index ones. 
This is due to the 
solvability of the recursive equation system~\eqref{matrix}  with the LDL 
decomposition~\eqref{LDLdecomp} in any ordering of the rows and columns of
$\Siga$. 

However, in the presence of an RZP in the undirected model, 
the perfect ordering of the nodes produces a
directed Gaussian graphical model without sink V configuration. This is
Markov equivalent to the undirected one, and so, a recursive ordering of
the variables exists in which
$\beta_{ij \cdot \{ i+1 ,\dots ,d \} }$, the  partial regression 
coefficient of $X_j$ when  regressing $X_i$
with $X_{i+1} ,\dots ,X_d $, is zero exactly when 
$r_{ij \cdot \{ 1 ,\dots ,d \} \setminus \{ i,j \}  }$ (the  partial correlation 
coefficient between $X_i$ and $X_j$, excluding the effect of all other
variables) is zero. It does not mean that they have the same value, but they
are zeros at the same time.
In this case,
the positions of the zero entries in the upper diagonal part of the 
concentration matrix  are identical 
to the  positions of the zero entries in $\A$. 
This is the special case of Corollary~\ref{moral} for Gaussian graphical models.

\end{appendices}


\begin{thebibliography}{999}
\bibitem[Abdelkhalek and Bolla(2020)]{Abdelkhalek}
Abdelkhalek, F., Bolla, M., Application of Structural Equation Modeling to Infant Mortality Rate in Egypt (2020). In: \textit{Demography of Population Health, Aging and Health Expenditures}, ed. Skiadas C.H., Skiadas C., Springer, Cham, 
pp. 89--99.






\bibitem[Akbilgic et al. (2014)]{Akbilgic}
Akbilgic, O., Bozdogal, H.,  Balaban, M. E.,
A Novel Hybrid RBF Neural Networks Model as a Forecaster,
\textit{Statistics and Computing} \textbf{24} (2014), 365-375.

\bibitem[Bazinas and Nielsen (2022)]{Bazinas}
Bazinas, V., Nielsen, B., Causal Transmission in Reduced-Form Models,
\textit{Econometrics} (2022), 10,14. 


\bibitem[Bolla et al. (2019)]{Bollacta}
Bolla, M., Abdelkhalek, F., Baranyi, M., Graphical models, regression graphs,
and recursive linear regression in a unified way, \textit{Acta Scientiarum
Mathematicarum (Szeged)} \textbf{85} (2019), 9--57.


\bibitem[Bolla and Szabados (2021)]{Szabados}
Bolla, M., Szabados, T., \textit{Multidimensional Stationary Time Series: 
Dimension
Reduction and Prediction}. CRC Press, Taylor and Francis Group (2021).

\bibitem[Box et al. (2015)]{Box}
Box, G.~E.~P., Jenkins, G.~M., Reinsel, G.~C. and Ljung, G.~M.,
\textit{Time series Analysis: Forecasting and Control}. Wiley (2015).

\bibitem[Brillinger (1996)]{Brillinger}
Brillinger, D. R., Remarks concerning graphical models for time series and
point processes, \textit{Revista de Econometria} 16(1) (1996), 1--23.

\bibitem[Brockwell and Davis (1991)]{Brockwell}
Brockwell, P. J., Davis, R. A., \textit{Time Series: Theory and Methods}.
Springer (1991). 


\bibitem[Deistler and Scherrer (2019)]{Deistler}
Deistler, M. and Scherrer, W., Vector Autoregressive Moving Average Models.
In \textit{Handbook of Statistics}, Vol. 41. Springer (2019).

\bibitem[Deistler and Scherrer (2022)]{Deistlerbook}
Deistler, M. and Scherrer, W., Time Series Models, Springer Nature (2022).

\bibitem[Dempster (1972)]{Dempster}
Dempster, A. P., Covariance selection, \textit{Biometrics} \textbf{28}
(1972), 157--175.

\bibitem[Eichler (2012)]{Eichler}
Eichler, M., Graphical modelling of multivariate time series,
\textit{Prob. Theory Related Fields} \textbf{153} (2012), 233--268.


\bibitem[Eichler (2016)]{Eichler1}
Eichler, M., Graphical modelling of dynamic relationships in multivariate time 
series. In B.S.J.T. M. Winterhalder ed., \textit{Handbook of Time Series 
Analysis}.
Wiley-VCH Berlin (2006).


\bibitem[Geweke (1984)]{Geweke}
Geweke, J., Inference and causality in economic time series models. In
\textit{Handbook of Econometrics}, Vol. 2. 
Elsevier (1984).

\bibitem[Granger (1969)]{Granger}
Granger, C., Investigating causal relations by econometric models and
cross-spectral methods, \textit{Econometrica} \textbf{37}(3) (1969), 424-438.

\bibitem[Golub and Van Loan (2012)]{Golub}
Golub, G. H. and Van Loan, C. F., \textit{Matrix computations}. 
JHU Press (2012).

\bibitem[Haavelmo (1943)]{Haavelmo}
T. Haavelmo, The statistical implications of a system of simultaneous equations,
\textit{Econometrica} \textbf 11 (1943), 1--12.

\bibitem[Joreskog (1977)]{Joreskog}
J\"oreskog, K. G., Structural equation models in the social sciences.
Specification, estimation and testing (1977). 
In \textit{Applications of statistics},
ed. P.R. Krishnaiah. North-Holland Publishing Co., 265--287.

\bibitem[Keating (1996)]{Keating}
Keating, J. W., Structural information is recursive VAR orderings,
\textit{J. of Econometric Dynamics and Control} \textbf{20} (1996), 1557-1580.

\bibitem[Kiiveri et al. (1984)]{Kiiveri}
Kiiveri, H., Speed, T. P., Carlin, J. B., Recursive casual models,
\textit{J. Austral. Math. Soc. (Ser. A)} \textbf{36} (1984), 30-52. 

\bibitem[Kilian and L\"utkepohl (2017)]{Kilian}
Kilian, L., L\"utkepohl, H., \textit{Structural Vector Autoregressive Analysis}.
Cambridge University Press (2017).

\bibitem[Lauritzen (2004)]{Lauritzen}
Lauritzen, S. L., \textit{Graphical Models}. Oxford statistical science 
series. Clarendon Pr., Oxford [u.a.], reprint with corr. edition (2004). 

\bibitem[L\"utkepohl (2005)]{Lutkepohl}
L\"utkepohl, H., \textit{New Introduction to Multiple Time Series Analysis}.
Springer, Berlin (2005).

\bibitem[Nocedal and Wright (1999)]{Nocedal}
Nocedal, J. and Wright, S. J., \textit{Numerical Optimization}. Springer (1999).

\bibitem[Rao (1973)]{Rao}
Rao, C. R., \textit{Linear Statistical Inference and its Applications}. 
Wiley (1973).

\bibitem[R\'ozsa (1991)]{Rozsa}
R\'ozsa, P., \textit{Linear algebra and its applications} (in Hungarian).
M\H{u}szaki Kiad\'o, Budapest (1991).

\bibitem[Sims (1980)]{Sims}
Sims, C.~A., Macroeconomics and reality, \textit{Econometrica} 
\textbf{48} (1) (1980), 1-48.

\bibitem[Tarjan and Yannakakis (1984)]{Tarjan}
Tarjan, R. E. and Yannakakis, M.,
Simple Linear-Time Algorithms to Test Chordality of Graphs, Test Acyclicity of Hypergraphs, and Selectively Reduce Acyclic Hypergraphs,
\textit{SIAM Journal on Computing}
\textbf{13} (3) (1984), 566--579.

\bibitem[Wainwright and Jordan (2008)]{WainwrightJordan}
Wainwright, M. J. and Jordan, M.I.,  Graphical models, exponential families, and
variational inference. Foundations and Trends in Machine Learning 1 (1-2),
1-305 (2008).

\bibitem[Wainwright (2015)]{Wainwright}
Wainwright, M. J., Graphical Models and Message-Passing Algorithms:
Some Introductory Lectures. In: Mathematical Foundations of Complex
Networked Information Systems, Lecture Notes in Mathematics 2141,
F. Fagnani et al. (eds.), Springer (2015).


\bibitem[Wermuth (1980)]{Wermuth}
Wermuth, N., Recursive equations, covariance selection, and path analysis,
\textit{J. Amer. Stat. Assoc.} \textbf{75} (1980), 963--972.

\bibitem[Wermuth (2015)]{Wer15}
Wermuth, N., Graphical Markov models, unifying results and their interpretation,
In: (Balakrishnan, N. et al. eds) Wiley StatsRef: Statistics Reference Online
(2015), also ArXiv: 1505.02456.

\bibitem[Wiener (1956)]{Wiener}
Wiener, N., The theory of prediction. In: E.~F.~Beckenback ed., \textit{Modern
Mathematics for Engineers}, McGraw--Hill, New York (1956).

\bibitem[Wold (1960)]{Wold}
Wold, H.~O.~A., A generalization of causal chain models, \textit{Econometrica}
\textbf{28} (2) (1960), 444--463.

\bibitem[Wold (1985)]{Wold1}
Wold, H.~O.~A., Partial least squares. In \textit{Encyclopedia of 
Statistical Sciences}, ed. Kotz~S, Johnson~N~L, Read~C~R. Wiley, 1985.

\bibitem[Wright (1934)]{Wright}
Wright, S., The method of path coefficients, \textit{Ann. Math. Stat.} 
\textbf{5} (3) (1934), 161--215.


\end{thebibliography}
\end{document}